\documentclass[a4paper,USenglish]{lipics}
\usepackage[utf8]{inputenc}
\usepackage{amsmath,amssymb,bbm,xspace,tikz}
\usetikzlibrary{shapes.misc}
\title{Towards a Characterization of Leaf Powers by Clique Arrangements}
\titlerunning{Towards a Characterization of Leaf Powers by Clique Arrangements} 

\author[1]{Ragnar Nevries}
\author[2]{Christian Rosenke}
\affil[1]{Department of Computer Science, University of Rostock\\
\texttt{ragnar.nevries@uni-rostock.de}}
\affil[2]{Department of Computer Science, University of Rostock\\
\texttt{christian.rosenke@uni-rostock.de}}
\authorrunning{R. Nevries and C. Rosenke}

\Copyright{R. Nevries and C. Rosenke}
 %please refer to \url{http://www.acm.org/about/class/ccs98-html} E.g., cite as "F.1.1 Models of Computation". 
% Please provide 1-5 keywords

\def\copyrightline{}

\newtheorem{clai}{Claim}

\newcommand{\etal}{et\,al.\ }
\newcommand{\nad}{\ensuremath{{\mid}}}
\newcommand{\ad}{\mathord{-}}
\newcommand{\dad}{\mathord{\rightarrow}}

\newcommand{\N}{\ensuremath{\mathbbm{N}}}

\newcommand{\cA}{\ensuremath{{\cal A}}}

\newcommand{\cC}{\ensuremath{{\cal C}}}

\newcommand{\cE}{\ensuremath{{\cal E}}}

\newcommand{\cL}{\ensuremath{{\cal L}}}

\newcommand{\cX}{\ensuremath{{\cal X}}}

\begin{document}

\maketitle

\begin{abstract}
The class $\cL_k$ of $k$-leaf powers consists of graphs $G=(V,E)$ that have a $k$-leaf root, that is, a tree $T$ with leaf set $V$, where $xy \in E$, if and only if the $T$-distance between $x$ and $y$ is at most $k$.
Structure and linear time recognition algorithms have been found for $2$-, $3$-, $4$-, and, to some extent, $5$-leaf powers, and it is known that the union of all $k$-leaf powers, that is, the graph class $\cL = \bigcup_{k=2}^\infty \cL_k$, forms a proper subclass of strongly chordal graphs.
Despite from that, no essential progress has been made lately.

In this paper, we use the new notion of clique arrangements to suggest that leaf powers are a natural special case of strongly chordal graphs.
The clique arrangement $\cA(G)$ of a chordal graph $G$ is a directed graph that represents the intersections between maximal cliques of $G$ by nodes and the mutual inclusion of these vertex subsets by arcs.
Recently, strongly chordal graphs have been characterized as the graphs that have a clique arrangement without bad $k$-cycles for $k \geq 3$.
We show that the clique arrangement of every graph of $\cL$ is free of bad $2$-cycles.
The question whether this characterizes the class $\cL$ exactly remains open.
\end{abstract}

\section{Introduction}

Leaf powers are a family of graph classes that has been introduced by Nishimura \etal \cite{NisRagThi2002} to model the problem of reconstructing phylogenetic trees.
In particular, a given finite simple graph $G=(V,E)$ is called the $k$-leaf power of a tree $T$ for some $k \geq 2$, if $V$ is the set of leaves in $T$ and any two distinct vertices $x,y \in V$ are adjacent, that is $xy \in E$, if and only if the distance of $x$ and $y$ in $T$ is at most $k$.
For all $k \geq 2$, the class of graphs that are a $k$-leaf power of some tree, is simply called $k$-leaf powers and denoted by $\cL_k$.
The general problem, from a graph theoretic point of view, is to structurally characterize $\cL_k$ for all fixed $k \geq 2$ and to provide efficient recognition algorithms.

Obviously, a graph $G$ is a $2$-leaf power, if and only if it is the disjoint union of cliques, that is, $G$ does not contain a chordless path of length $2$.
%TODO: check references
Dom \etal \cite{DomGuoHueNie2004,DomGuoHueNie2005} prove that $3$-leaf powers are exactly the graphs that do not contain an induced bull, dart, or gem.
Brandstädt \etal \cite{BraLe2006} contribute to the characterization of $3$-leaf powers by showing that they are exactly the graphs that result from substituting cliques into the nodes of a tree.
Moreover, they give a linear time algorithm to recognize $3$-leaf powers building on their characterization.
A characterization of $4$-leaf powers in terms of forbidden subgraphs is yet unknown.
However, basic $4$-leaf powers, the $4$-leaf powers without true twins, are characterized by eight forbidden subgraphs \cite{Raute2006}. 
The structure of basic $4$-leaf powers has further been analyzed by Brandstädt \etal \cite{BraLeSri2006}, who provide a nice characterization of the two-connected components of basic $4$-leaf powers that leads to a linear time recognition algorithm even for $4$-leaf powers.
For $5$-leaf powers, a polynomial time recognition was given in \cite{ChaKo2007}.
However, no structural characterization is known, even for basic $5$-leaf powers. 
Only for distance-hereditary basic $5$-leaf powers a characterization in terms of $34$ forbidden induced subgraphs has been discovered \cite{BraLeRau2006}.
Except from the result in \cite{BraWag2008} that $\cL_k \subseteq \cL_{k+1}$ is not true for every $k$, there have not been any more essential advances in determining the structure of $k$-leaf powers for $k \geq 5$ since 2007.
Instead, research has focused on generalizations of leaf powers \cite{BraLe2008,BraWag2007}, which also turned into dead ends, very soon.

On the other hand, if we push $k$ to infinity, then it turns out that not every graph is a $k$-leaf power for some $k \geq 2$.
In particular, a $k$-leaf power is, by definition, the subgraph of the $k$-th power of a tree $T$ induced by the leaves of $T$.
Since trees are sun-free chordal and as taking powers and induced subgraphs do not destroy this property, it follows trivially that every $k$-leaf power, despite the value of $k$, is strongly chordal \cite{Far1983}.
But even not every strongly chordal graph is a $k$-leaf power for some $k \geq 2$.
In fact, we are aware of exactly one counter example, which has been found by Bibelnieks \etal \cite{BibDea1993} and is shown as $G_7$ in Figure \ref{fig:counterexamples}.
Insofar, it is reasonable to ask for a precise characterization of the graphs that are not a $k$-leaf power for any $k \geq 2$.
This problem can equivalently be formulated as to describe the graphs in the class $\cL = \bigcup_{k = 2}^\infty \cL_k$, which we call leaf powers, for short.

Interestingly, Brandstädt \etal \cite{BraHunManWag2009} show that $\cL$ coincides with the class of fixed tolerance NeST (neighborhood subtree tolerance) graphs, a well-known graph class with an absolutely different motivation given by Bibelnieks \etal \cite{BibDea1993}.
Naturally, characterizations and an efficient recognition algorithms for this class are also open questions today.
However, by Brandstädt \etal \cite{BraHun2008,BraHunManWag2009}, it is know that $\cL$ is a superclass of ptolemaic graphs, that is, gem-free chordal graphs \cite{How1981}, and even a superclass of directed rooted path graphs, introduced by Gavril \cite{Gav1974}.

Recently, we introduced the clique arrangement in \cite{NevRos2013}, a new data structure that is especially valuable for the analysis of strongly chordal graphs.
The clique arrangement $\cA(G)=(\cX,\cE)$ of a chordal graph $G$ is a directed acyclic graph that has certain vertex subsets of $G$ as a node set and describes the mutual inclusion of these sets by arcs.
In particular, for every set $C_1, C_2, \ldots$ of maximal cliques of $G$ there is a node in $\cX$ for $X = C_1 \cap C_2 \cap \ldots$ and two nodes $X,Z \in \cX$ are joined by an arc $XZ \in \cE$, if $X \subset Z$ and there is no $Y \in \cX$ with $X \subset Y \subset Z$.
In \cite {NevRos2013}, we give a new characterization of strongly chordal graphs in terms of a forbidden cyclic substructure in the clique arrangement, called bad $k$-cycles for $k\ge 3$, and we show how to construct the clique arrangement of a strongly chordal graph in nearly linear time.

It is known that the clique arrangements of ptolemaic graphs are even directed trees \cite{UehUno2005}.
Since all ptolemaic graphs are leaf powers and all leaf powers are strongly chordal, it appears likely that the degree of acyclicity in clique arrangements of leaf powers is between forbidden bad $k$-cycles, $k \ge 3$, and the complete absence of cycles.

This paper describes a cyclic substructure that is forbidden in the clique arrangement of leaf powers.
For convenience, we call these substructures bad $2$-cycles, although they are not the obvious continuation of the concept of bad $k$-cycles for $k \ge 3$.
As the main result of this paper, we show that bad $2$-cycles occur in $\cA(G)$, if and only if $G$ contains at least one of seven induced subgraphs $G_1, \ldots, G_7$ depicted in Figure \ref{fig:counterexamples}.

We leave it as an open question, if these seven graphs are sufficient to characterize $\cL$ in terms of forbidden subgraphs.
However, we conjecture that this is the case.
This would imply a polynomial time recognition algorithm for $\cL$, by using the possibility of efficiently recognizing strongly chordal graphs and checking the containment of a finite number of forbidden induced subgraphs.

\section{Preliminaries}

We refer to several graph classes which are not explicitly defined due to space limitations.
For a comprehensive survey on graph classes we would like to refer to \cite{BraLeSpi1999}.

Throughout this paper, all graphs $G=(V,E)$ are simple, without loops and, with the exception of clique arrangements, undirected.
We usually denote the vertex set by $V$ and the edge set by $E$, where the edges are also called arcs in a directed graph.
We write $x \ad y$, respectively $x \dad y$ in the directed case, for $xy \in E$ and $x \nad  y$ for $xy \not\in E$.
For all vertices $x \in V$ in an undirected graph, we let $N(x)=\{y\ |\ xy \in E\}$ denote the \emph{open neighborhood} and $N[x] = N(x) \cup \{x\}$ the \emph{closed neighborhood} of $x$ in $G$.
In a directed graph, $N_o(x)=\{y\ |\ xy \in E\}$ denotes the set of neighbors that are reachable from $x$ by a single arc and $N_i(x)=\{y\ |\ yx \in E\}$ are the neighbors that reach $x$ by a single arc.
If $|N_i(x)| = 0$ then $x$ is a \emph{source} and if $|N_o(x)| = 0$ then $x$ is a \emph{sink}.

An \emph{independent set} in $G$ is a set of mutually nonadjacent vertices.
A \emph{clique} $C \subseteq V$ is a set of mutually adjacent vertices and $C$ is called maximal, if there is no clique $C'$ with $C \subset C'$.
The set of all maximal cliques of $G$ is denoted by $\cC(G)$.

%For any subset $U \subseteq V$, we let $G[U]=(U,E_U)$ denote the subgraph of $G$ induced by $U$.
%We simply write $G-U$ for $G[V\setminus U]$ and $G-x$ for $G-\{x\}$. %TODO used?

A (simple) path in a graph $G$ is a sequence $x_1, x_2 \ldots, x_k$ of non-repeating vertices in $G$, such that $x_ix_{i+1} \in E$ for all $i \in \{1, \ldots, k-1\}$.
If $E$ is clear from the context, then we denote the path by $x_1 \ad x_2 \ad \ldots \ad x_k$ in an undirected graph.
In a directed graph, $x_1 \dad x_2 \dad \ldots \dad x_k$ specifies a directed path and we say that $x_1$ \emph{reaches} $x_k$.
The distance $d_G(x,y)$ between two vertices $x,y$ of an (un-) directed graph $G$ is the minimum number of edges in an (un-) directed path starting in $x$ and ending in $y$.
If the edge $x_kx_1$ is additionally present in $E$, then we talk of a (simple) cycle in $G$, and as for paths, an undirected cycle is denoted by $x_1 \ad x_2 \ad \ldots \ad x_k \ad x_1$.
An undirected cycle is called \emph{induced $k$-cycle $C_k$}, if $G$ contains $x_ix_j$, if and only if $j = i+1$ or $i=k$ and $j=1$.

A tree $T$ is an undirected connected acyclic graph, that is, for all pairs $x,y$ of vertices there exists a path $x \ad \ldots \ad y$, and $T$ is free of cycles.
Directed graphs are acyclic, if they are free of directed cycles.

A vertex subset $U = \{x_0, \ldots, x_{k-1}, y_0, \ldots, y_{k-1}\} \subseteq V$ induces a \emph{$k$-sun} in $G$, if $X=\{x_0, \ldots, x_{k-1}\}$ is a clique and $Y=\{y_0, \ldots, y_{k-1}\}$ is an independent set and for every edge $x_iy_j$ between $X$ and $Y$, either $i=j$ or $i+1=j$, where the indices are counted modulo $k$.
By definition, a graph is \emph{chordal}, if and only if it does not contain induced $k$-cycles for all $k \geq 4$, and by Farber \cite{Far1983} a graph is \emph{strongly chordal}, if and only if it does not contain induced $k$-suns for all $k \geq 3$.

Beside the many useful properties of (strongly) chordal graphs, see for example \cite{BraLeSpi1999}, this paper uses in particular the following two properties, that are folklore but nevertheless have been shown in \cite{NevRos2013}:
\begin{lemma}\label{lma_vertex_node_none_adjacency}
If $G$ is a chordal graph and $C_1,C_2$ are maximal cliques of $G$, then there is a vertex $x \in C_1\setminus C_2$ such that $x \nad y$ for all $y \in C_2\setminus C_1$.
\end{lemma}
\begin{lemma}\label{lma_stronglyChordal_cliqueCutStructure}
If $G$ is a strongly chordal graph and $\cC$ any nonempty subset of $\cC(G)$, then there are two maximal cliques $C_1, C_2 \in \cC$ such that $\bigcap_{C \in \cC} C = C_1 \cap C_2$.
\end{lemma}

A strongly chordal graph $G=(V,E)$ is the $k$-leaf power of a tree $T$ for $k \geq 2$, if $V$ is the set of leaves in $T$ and for all $x,y \in V$ there exists $xy \in E$, if and only if $d_T(x,y) \leq k$.
The tree $T$ is called a $k$-leaf root of $G$, in this case.
Notice that $k$-leaf roots are not necessarily unique for given $k$-leaf powers.
For all $k \geq 2$, the class $\cL_k$ consists of all graphs that are a $k$-leaf power for some tree and $\cL = \bigcup_{k=2}^\infty \cL_k$ is the class of leaf powers.
%TODO: Leaf power stuff

The clique arrangement $\cA(G) = (\cX,\cE)$ of a chordal graph $G$, as introduced in \cite{NevRos2013}, is a directed acyclic graph with node set 
\[\cX = \left\{X\ \left|\ X = \bigcap_{C \in \cC} C \text{ with } \cC \subseteq \cC(G) \text{ and } X \not= \emptyset \right.\right\},\]
that contains exactly all intersections of the maximal cliques of $G$, and arc set
\[\cE = \left\{XZ\ \left|\ X,Z \in \cX \text{ with } X \subset Z \text{ and } \nexists Y \in \cX: X \subset Y \subset Z \right.\right\}\]
that describes their mutual inclusion.
Clearly, the set of sinks in $\cA(G)$ corresponds exactly to $\cC(G)$.

The following simple facts for clique arrangements are also introduced in \cite{NevRos2013}:
\begin{lemma}[Nevries and Rosenke \cite{NevRos2013}]\label{obs_sink_intersection}
If $X \in \cX$ is a node in the clique arrangement $\cA(G)=(\cX,\cE)$ of a chordal graph $G$ and if $\{Y_1, \ldots, Y_\ell\} = N_o(X)$, then $X = Y_1 \cap \ldots \cap Y_\ell$.
Moreover, if $C_1, \ldots, C_k$ are the sinks of $\cA(G)$ that are reached from $X$ by directed paths, then $X = C_1 \cap \ldots \cap C_k$.
\end{lemma}
\begin{lemma}[Nevries and Rosenke \cite{NevRos2013}]\label{obs_node_intersection}
If $Y_1, \ldots, Y_k \in \cX$ are nodes in the clique arrangement $\cA(G)=(\cX,\cE)$ of a chordal graph $G$ such that their intersection $X = Y_1 \cap \ldots \cap Y_k$ is not empty, then $X \in \cX$.
\end{lemma}

Although $\cA(G)$ is acyclic by definition, we call the following structure a cycle in $\cA(G)$ for the lack of a better term. 
For any $k \in \N$, a \emph{$k$-cycle} of $\cA(G)$ is a set of nodes $S_0, \ldots, S_{k-1}, T_0, \ldots, T_{k-1}$ such that for all $i \in \{0, \ldots, k-1\}$ there is a directed path from $S_i$ to $T_i$ and a directed path from $S_i$ to $T_{i-1}$ (counted modulo $k$).
The nodes $S_0, \ldots, S_{k-1}$ are called \emph{starters} of the cycle and the nodes $T_0, \ldots, T_{k-1}$ are called \emph{terminals} of the cycle.
Note that by definition, $S_i\subseteq T_i \cap T_{i-1}$ for all $i \in \{0, \ldots, k-1\}$.
In \cite{NevRos2013}, we call a $k$-cycle \emph{bad}, if $k \geq 3$ and for all $i,j \in \{0, \ldots, k-1\}$ there is a directed path from $S_i$ to $T_j$, if only if $j \in \{i, i-1\}$ (counted modulo $k$).
\begin{theorem}[Nevries and Rosenke \cite{NevRos2013}]\label{thm_badcycles}
Let $G=(V,E)$ be a chordal graph and $\cA(G)=(\cX,\cE)$ be the clique arrangement of $G$.
Then $G$ is strongly chordal, if and only if $\cA(G)$ is free of bad $k$-cycles for all $k \geq 3$.
\end{theorem}
In this paper we apply two other properties of clique arrangements for strongly chordal graphs:
\begin{lemma}[Proof in Section \ref{sec:techproofs}]\label{lma_node_to_sink_intersection}
Let $G$ be a strongly chordal graph with clique arrangement $\cA(G)=(\cX,\cE)$ and let $X,Y,Z \in \cX$ be three distinct nodes such that $X = Y \cap Z$.
There are sinks $C_1,C_2 \in \cX$ such that $C_1$ is reachable from $Y$ and $C_2$ is reachable from $Z$ and $X = C_1 \cap C_2$.
\end{lemma}
\begin{lemma}[Proof in Section \ref{sec:techproofs}]\label{lma_cliqueArrangement_inducedSubgraphs}
Let $G=(V,E)$ be a chordal graph with clique arrangement $\cA(G)=(\cX,\cE)$ that occurs as an induced subgraph of a chordal graph $G'=(V', E')$ with clique arrangement $\cA(G')=(\cX',\cE')$, that is, $G= G'[V]$.
There exists a function $\phi: \cX \rightarrow \cX'$ that fulfills the following two conditions for all $X,Y \in \cX$:
\begin{enumerate}
\item $X = Y \Leftrightarrow \phi(X) = \phi(Y)$, and
\item $\cA(G)$ has a directed path from $X$ to $Y$, if and only if $\cA(G')$ has a directed path from $\phi(X)$ to $\phi(Y)$.
\end{enumerate}
\end{lemma}

\section{Forbidden Induced Subgraphs}

Bibelnieks \etal \cite{BibDea1993} are the first to find a strongly chordal graph, namely $G_7$, that is not in $\cL$ and, consequently, show that the classes are not equivalent.
In fact, they were looking for a strongly chordal graph that is not a fixed tolerance NeST graph, but by Brandstädt \etal \cite{BraHunManWag2009}, we know that $\cL$ and this class are equal.
Since then, it has been conjectured that $G_7$ is the smallest forbidden induced subgraph of leaf powers.

To show that $G_7$ is not in $\cL$, Bibelnieks \etal \cite{BibDea1993} use a lemma of Broin \etal \cite{BroLow1986}.
The basic idea of the proof of this lemma is to show for certain pairs of edges $x_1y_1$ and $x_2y_2$ in $G$ that the path between $x_1$ and $y_1$ is disjoint from the path between $x_2$ and $y_2$ in every leaf root of $G$.
In particular, this happens, if vertices $a,b$ exist in $G$ with $x_1,y_1 \in N(a) \setminus N[b]$ and $x_2,y_2 \in N(b) \setminus N[a]$. 
The graph $G_7$ has a cycle $x_0 \ad y_{00} \ad y_{10} \ad x_1 \ad y_{11} \ad y_{01} \ad x_0$, where the condition is fulfilled for many pairs of edges in the cycle.
It follows that every leaf root of $G_7$ would have a cycle, which is a contradiction.

In this section, we want to show that there are at least six other strongly chordal graphs $G_1, \ldots, G_6$ that are not in $\cL$.
Interestingly, every of these six graphs is smaller than $G_7$.
For our proof, we generalize the argument of Bibelnieks \etal \cite{BibDea1993} for pairs of edges $x_1y_1$ and $x_2y_2$ that correspond to disjoint paths in leaf roots.
The following Lemma provides three corresponding conditions:
\begin{lemma}\label{lma_criticalEdges}
Let $G=(V,E)$ be a $k$-leaf power of a tree $T$ for some $k \geq 2$ and let $x_1y_1$ and $x_2y_2$ be two edges of $G$ on distinct vertices $x_1,y_1,x_2,y_2 \in V$.
The paths $x_1 \ad \ldots \ad y_1$ and $x_2 \ad \ldots \ad y_2$ in $T$ are disjoint, that is, do not share any node, if at least one of the following conditions holds:
\begin{enumerate}
\item
At most one of the edges $x_1x_2,x_1y_2,y_1x_2,y_1y_2$ is in $E$.
\item
There is a vertex $a \in V$ such that $x_1,y_1 \in N(a)$, and $x_2,y_2 \not\in N[a]$, and $N(x_1) \cap \{x_2,y_2\} \leq 1$, and $N(y_1) \cap \{x_2,y_2\} \leq 1$.
\item
There are distinct vertices $a,b \in V$ such that $x_1,y_1 \in N(a)\setminus N[b]$, and $x_2,y_2 \in N(b)\setminus N[a]$.
\end{enumerate}
\end{lemma}
\begin{proof}
\ \newline
\noindent 1. Assume that the two paths are not disjoint.
Then $T$ contains (not necessarily distinct) nodes $s$ and $t$ such that (i) the path $x_1 \ad \ldots \ad y_1$ consists of three subpaths, firstly $x_1 \ad \ldots \ad s$, secondly $s \ad \ldots \ad t$, and thirdly $t \ad \ldots \ad y_1$ and (ii) the path $x_2 \ad \ldots \ad y_2$ consists of three subpaths, too, without loss of generality,  the first is $x_2 \ad \ldots \ad  s$ and the last is $t \ad \ldots \ad y_2$.
Hence, the path between $s$ and $t$ is the intersection between the two paths.
Because $x_1 \ad y_1$ and $x_2 \ad y_2$ in $G$ we get the following inequations by definition:
\begin{align}
d_T(x_1,y_1) &= d_T(x_1,s) + d_T(s,t) + d_T(t,y_1) \leq k \text{ and} \label{eqn_lem_criticalEdges_a}\\
d_T(x_2,y_2) &= d_T(x_2,s) + d_T(s,t) + d_T(t,y_2) \leq k.\label{eqn_lem_criticalEdges_b}
\end{align}
As at most one of the edges $x_1x_2,x_1y_2,y_1x_1,y_1y_2$ is in $E$, we know that at least one of $x_1y_2,y_1x_2 \not\in E$ and $x_1x_2,y_1y_2 \not\in E$ is true.
If $x_1 \nad y_2$ and $y_1 \nad x_2$, then we get
\begin{align}
d_T(x_1,y_2) &= d_T(x_1,s) + d_T(s,t) + d_T(t,y_2) > k\text{ and} \label{eqn_lem_criticalEdges_c}\\
d_T(y_1,x_2) &= d_T(y_1,t) + d_T(t,s) + d_T(s,x_2) > k\label{eqn_lem_criticalEdges_d}
\end{align}
such that combining (\ref{eqn_lem_criticalEdges_a}) and (\ref{eqn_lem_criticalEdges_c}) yields $d_T(t,y_1) < d_T(t,y_2)$ and combining (\ref{eqn_lem_criticalEdges_b}) and (\ref{eqn_lem_criticalEdges_d}) yields $d_T(t,y_2) < d_T(t,y_1)$, a contradiction.
Otherwise, if $x_1 \nad x_2$ and $y_1 \nad y_2$, we get the inequations
\begin{align}
d_T(x_1,x_2) &= d_T(x_1,s) + d_T(s,x_2) > k\text{ and} \label{eqn_lem_criticalEdges_e}\\
d_T(y_1,y_2) &= d_T(y_1,t) + d_T(t,y_2) > k\label{eqn_lem_criticalEdges_f}
\end{align}
such that combining equation (\ref{eqn_lem_criticalEdges_a}) and (\ref{eqn_lem_criticalEdges_e}) yields $d(x_2,s) > d_T(s,t) + d_T(t,y_1)$.
Putting this estimate of $d_T(x_2, s)$ into (\ref{eqn_lem_criticalEdges_b}) yields $d_T(y_1,t) + d_T(t,y_2) + 2d_T(s,t) < k$.
By (\ref{eqn_lem_criticalEdges_f}) we can conclude that $2d_T(s,t) < 0$, which is a contradiction to the preconditions.

\smallskip
\noindent 2. As the edges $ax_1$ and $x_2y_2$ are joined in $G$ by at most one edge, $x_1x_2$ or $x_1y_2$, it follows from 1.\ that $a \ad \ldots \ad x_1$ is disjoint from $x_2 \ad \ldots \ad y_2$ in $T$.
Analogously, the edges $ay_1$ and $x_2y_2$ are joined by at most one edge in $G$, either $y_1x_2$ or $y_1y_2$.
Hence, in $T$, the path $a \ad \ldots \ad y_1$ is disjoint from $x_2 \ad \ldots \ad y_2$, too.
Because $T$ is a tree, it follows that the nodes on $x_1 \ad \ldots \ad y_1$ are a subset of the combined nodes of the paths $a \ad \ldots \ad x_1$ and $a \ad \ldots \ad y_1$.
Consequently, there is no node that simultaneously belongs to $x_1 \ad \ldots \ad y_1$ and $x_2 \ad \ldots \ad y_2$.

\smallskip
\noindent 3. If $a \ad b$ then $z_1 \nad z_2$ for all $z_1 \in \{x_1,y_1\}$ and $z_2 \in \{x_2,y_2\}$.
Otherwise, $z_1 \ad a \ad b \ad z_2 \ad z_1$ is an induced $C_4$ in $G$.
Hence, in this case $x_1 \nad x_2$, $x_1 \nad y_2$, $y_1 \nad x_2$ and $y_1 \nad y_2$ and we are done.

If $a \nad b$, then for all $z_1 \in \{x_1,y_1\}$ and $z_2 \in \{x_2,y_2\}$, the edges $a \ad z_1$ and $b \ad z_2$ are joined at most by the edge $z_1 \ad z_2$ in $G$.
This means by 1.\ that $a \ad \ldots z_1$ is disjoint from $b \ad \ldots \ad z_2$ in $T$.
Again, as $T$ is a tree, it follows that the nodes on $x_1 \ad \ldots \ad y_1$ are a subset of the accumulated nodes on $a \ad \ldots \ad x_1$ and $a \ad \ldots \ad y_1$ and, similarly, the nodes on $x_2 \ad \ldots \ad y_2$ are a subset of the nodes on $b \ad \ldots \ad x_2$ and $b \ad \ldots \ad y_2$.
Consequently, there cannot be a node that simultaneously belongs to $x_1 \ad \ldots \ad y_1$ and $x_2 \ad \ldots \ad y_2$.
\end{proof}
Based on this more general concept, we can find a cycle $x_0 \ad y_{00} \ad y_{10} \ad x_1 \ad y_{11} \ad y_{01} \ad x_0$ in every graph from $G_1, \ldots, G_7$ such that many pairs of edges in the cycle fulfill at least one of the three conditions.
The following theorem states that this is never compatible with the existence of a leaf root.

\begin{theorem}[Proof in Section \ref{sec:techproofs}]\label{thm_counterexamples}
The graphs $G_1, \ldots, G_7$ are not in $\cL$.
\end{theorem}
\begin{figure}[ptb]
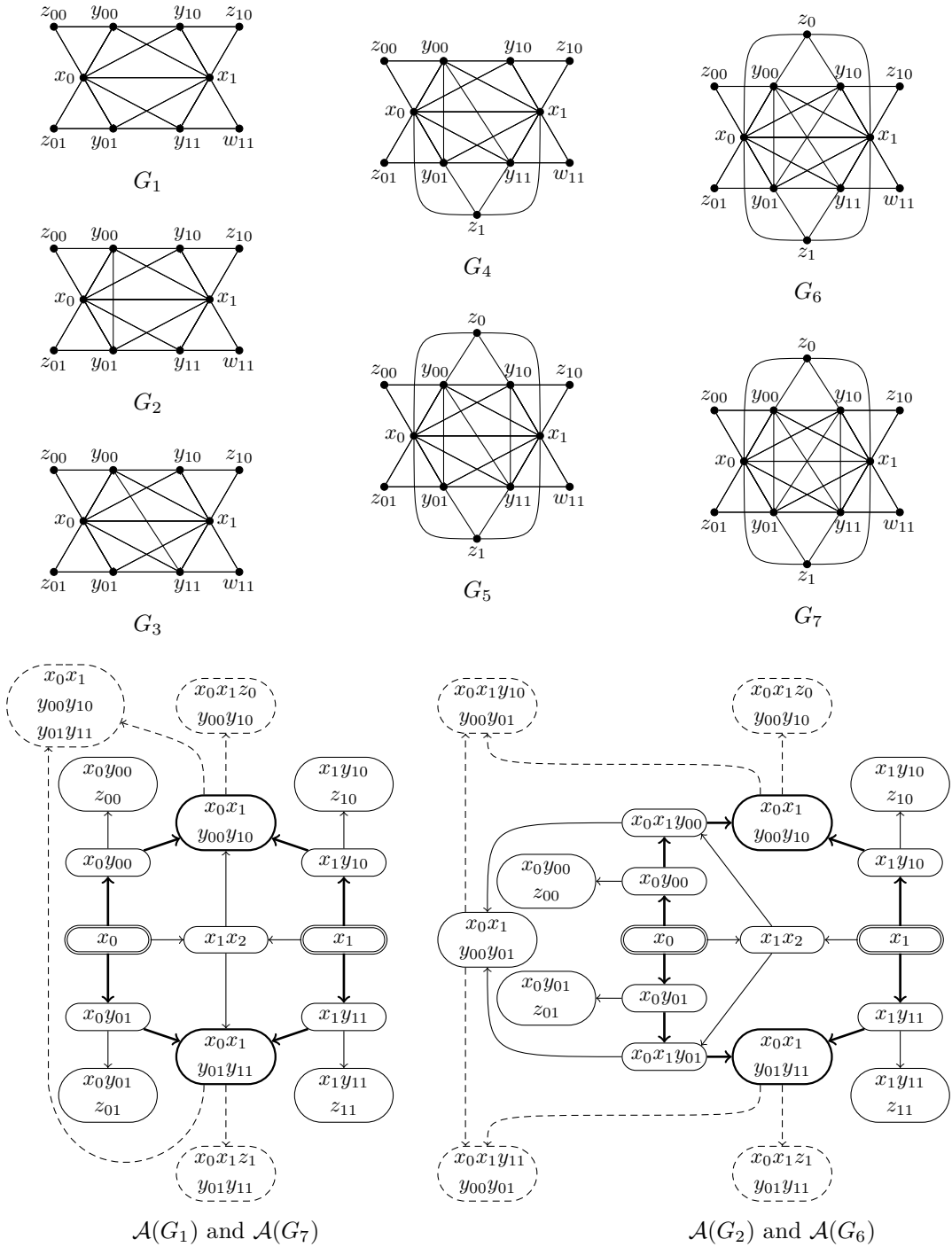

\tikzstyle{vt}=[draw=black,fill=black,circle,minimum size=3pt,inner sep=0]
\newcommand{\spread}{0.125}
\newcommand{\vspread}{0.2cm}
\newcommand{\radius}{1.75}
\newcommand{\tscale}{0.9}
\newcommand{\gscale}{0.88}
\newcommand{\clique}[1]{%TODO clique without double edges?
\foreach \x in {#1} {
\foreach \y in {#1} {
\draw (\x) edge (\y);
}}}
\newcommand{\counterBase}{%
\draw (0,0) +(180:1) node[vt] (u0) {};
\draw (0,0) +(120:1) node[vt] (v00) {};
\draw (0,0) +(240:1) node[vt] (v01) {};
\draw (v00) +(-1,0) node[vt] (w00) {};
\draw (v01) +(-1,0) node[vt] (w01) {};
\draw (\spread,0) +(0:1) node[vt] (u1) {};
\draw (\spread,0) +(60:1) node[vt] (v10) {};
\draw (\spread,0) +(-60:1) node[vt] (v11) {};
\draw (v10) +(1,0) node[vt] (w10) {};
\draw (v11) +(1,0) node[vt] (w11) {};
\clique{u0,v00,v10,u1}
\clique{u0,v01,v11,u1}
\clique{w00,u0,v00}
\clique{w01,u0,v01}
\clique{w10,u1,v10}
\clique{w11,u1,v11}
\draw (u0) node[left,scale=\tscale] {$x_0$};
\draw (w00) node[above,scale=\tscale] {$z_{00}$};
\draw (w01) node[below,scale=\tscale] {$z_{01}$};
\draw (u1) node[right,scale=\tscale] {$x_1$};
\draw (w10) node[above,scale=\tscale] {$z_{10}$};
\draw (w11) node[below,scale=\tscale] {$w_{11}$};
\draw (v00) node[anchor=300,scale=\tscale] {$y_{00}$};
\draw (v10) node[anchor=240,scale=\tscale] {$y_{10}$};
\draw (v01) node[anchor=60,scale=\tscale] {$y_{01}$};
\draw (v11) node[anchor=120,scale=\tscale] {$y_{11}$};
}
\begin{minipage}[t]{\textwidth}
\begin{minipage}[c]{0.3\textwidth}
\centering
\tikz[baseline={(0,0)},scale=\gscale]{
\counterBase
}\\[\vspread]
$G_1$\\[\vspread]
\vspace*{\vspread}
\tikz[baseline={(0,0)},scale=\gscale]{
\counterBase
\draw (v00) edge (v01);
}\\[\vspread]
$G_2$\\[\vspread]
\vspace*{\vspread}
\tikz[baseline={(0,0)},scale=\gscale]{
\counterBase
\draw (v00) edge (v11);
}\\[\vspread]
$G_3$
\end{minipage}\hfill
\begin{minipage}[c]{0.3\textwidth}
\centering
\tikz[baseline={(0,0)},scale=\gscale]{
\counterBase
\draw (v00) edge (v11);
\draw (v00) edge (v01);
\draw (0.5*\spread,-1.75) node[vt] (w1) {};
\draw (w1) edge (v01);
\draw (w1) edge (v11);
\path (w1) -| (u0) coordinate[pos=0.5] (c0);
\path (w1) -| (u1) coordinate[pos=0.5] (c1);
\draw (w1) .. controls (c0) .. (u0);
\draw (w1) .. controls (c1) .. (u1);
\draw (w1) node[below,scale=\tscale] {$z_1$};
}\\[\vspread]
$G_4$\\[\vspread]
\vspace*{\vspread}
\tikz[baseline={(0,0)},scale=\gscale]{
\counterBase
\draw (v00) edge (v11);
\draw (v00) edge (v01);
\draw (v10) edge (v11);
\draw (0.5*\spread,1.75) node[vt] (w0) {};
\draw (w0) edge (v00);
\draw (w0) edge (v10);
\path (w0) -| (u0) coordinate[pos=0.5] (c0);
\path (w0) -| (u1) coordinate[pos=0.5] (c1);
\draw (w0) .. controls (c0) .. (u0);
\draw (w0) .. controls (c1) .. (u1);
\draw (w0) node[above,scale=\tscale] {$z_0$};
\draw (0.5*\spread,-1.75) node[vt] (w1) {};
\draw (w1) edge (v01);
\draw (w1) edge (v11);
\path (w1) -| (u0) coordinate[pos=0.5] (c0);
\path (w1) -| (u1) coordinate[pos=0.5] (c1);
\draw (w1) .. controls (c0) .. (u0);
\draw (w1) .. controls (c1) .. (u1);
\draw (w1) node[below,scale=\tscale] {$z_1$};
}\\[\vspread]
$G_5$
\end{minipage}\hfill
\begin{minipage}[c]{0.3\textwidth}
\centering
\tikz[baseline={(0,0)},scale=\gscale]{
\counterBase
\draw (v00) edge (v11);
\draw (v00) edge (v01);
\draw (v01) edge (v10);
\draw (0.5*\spread,1.75) node[vt] (w0) {};
\draw (w0) edge (v00);
\draw (w0) edge (v10);
\path (w0) -| (u0) coordinate[pos=0.5] (c0);
\path (w0) -| (u1) coordinate[pos=0.5] (c1);
\draw (w0) .. controls (c0) .. (u0);
\draw (w0) .. controls (c1) .. (u1);
\draw (w0) node[above,scale=\tscale] {$z_0$};
\draw (0.5*\spread,-1.75) node[vt] (w1) {};
\draw (w1) edge (v01);
\draw (w1) edge (v11);
\path (w1) -| (u0) coordinate[pos=0.5] (c0);
\path (w1) -| (u1) coordinate[pos=0.5] (c1);
\draw (w1) .. controls (c0) .. (u0);
\draw (w1) .. controls (c1) .. (u1);
\draw (w1) node[below,scale=\tscale] {$z_1$};
}\\[\vspread]
$G_6$\\[\vspread]
\vspace*{\vspread}
\tikz[baseline={(0,0)},scale=\gscale]{
\counterBase
\draw (v00) edge (v11);
\draw (v00) edge (v01);
\draw (v10) edge (v11);
\draw (v01) edge (v10);
\draw (0.5*\spread,1.75) node[vt] (w0) {};
\draw (w0) edge (v00);
\draw (w0) edge (v10);
\path (w0) -| (u0) coordinate[pos=0.5] (c0);
\path (w0) -| (u1) coordinate[pos=0.5] (c1);
\draw (w0) .. controls (c0) .. (u0);
\draw (w0) .. controls (c1) .. (u1);
\draw (w0) node[above,scale=\tscale] {$z_0$};
\draw (0.5*\spread,-1.75) node[vt] (w1) {};
\draw (w1) edge (v01);
\draw (w1) edge (v11);
\path (w1) -| (u0) coordinate[pos=0.5] (c0);
\path (w1) -| (u1) coordinate[pos=0.5] (c1);
\draw (w1) .. controls (c0) .. (u0);
\draw (w1) .. controls (c1) .. (u1);
\draw (w1) node[below,scale=\tscale] {$z_1$};
}\\[\vspread]
$G_7$
\end{minipage}
\end{minipage}

\tikzstyle{myNode}=[shape=rounded rectangle, text width=1.1cm, text centered,draw=black,scale=\tscale]
\tikzstyle{myNodeOpt}=[shape=rounded rectangle, text width=1.1cm, text centered,draw=black, densely dashed,scale=\tscale]
\tikzstyle{myNodeSta}=[shape=rounded rectangle, text width=1.1cm, text centered,draw=black, double,scale=\tscale]
\tikzstyle{myNodeTer}=[shape=rounded rectangle, text width=1.1cm, text centered,draw=black, line width=0.85pt,scale=\tscale]
\tikzstyle{myEdge}=[color=black]
\tikzstyle{myEdgeOpt}=[color=black, densely dashed]
\tikzstyle{myEdgeCyc}=[color=black, line width=1pt]
\vspace*{\vspread}
\vspace*{\vspread}
\tikz[baseline={(0,0)}]{
\node[myNodeSta] (s0) at (180:\radius) {$x_0$};
\node[myNodeSta] (s1) at (0:\radius) {$x_1$};
\node[myNode] (t) at (0,0) {$x_1 x_2$}
  edge[myEdge,<-] (s0)
  edge[myEdge,<-] (s1);
\node[myNode] (p00) at (s0|-0,0.66*\radius) {$x_0 y_{00}$}
  edge[myEdgeCyc,<-] (s0);
\node[myNode] (p01) at (s0|-0,-0.66*\radius) {$x_0 y_{01}$}
  edge[myEdgeCyc,<-] (s0); 
\node[myNode] (p10) at (s1|-0,0.66*\radius) {$x_1 y_{10}$}
  edge[myEdgeCyc,<-] (s1);
\node[myNode] (p11) at (s1|-0,-0.66*\radius) {$x_1 y_{11}$}
  edge[myEdgeCyc,<-] (s1); 
\node[myNodeTer] (t0) at (90:\radius) {$x_0 x_1$ $y_{00} y_{10}$}
  edge[myEdge,<-] (t)
  edge[myEdgeCyc,<-] (p00)
  edge[myEdgeCyc,<-] (p10);
\node[myNodeTer] (t1) at (270:\radius) {$x_0 x_1$ $y_{01} y_{11}$}
  edge[myEdge,<-] (t)
  edge[myEdgeCyc,<-] (p01)
  edge[myEdgeCyc,<-] (p11);
\node[myNode] (q00) at (s0|-0,1.33*\radius) {$x_0 y_{00}$ $z_{00}$}
  edge[myEdge,<-] (p00);
\node[myNode] (q01) at (s0|-0,-1.33*\radius) {$x_0 y_{01}$ $z_{01}$}
  edge[myEdge,<-] (p01);
\node[myNode] (q10) at (s1|-0,1.33*\radius) {$x_1 y_{10}$ $z_{10}$}
  edge[myEdge,<-] (p10);
\node[myNode] (q11) at (s1|-0,-1.33*\radius) {$x_1 y_{11}$ $z_{11}$}
  edge[myEdge,<-] (p11);
\node[myNodeOpt] (t0p) at (t0|-0,2*\radius) {$x_0 x_1 z_0$ $y_{00} y_{10}$};
\draw[myEdgeOpt,->] (t0) -- (t0p);
\node[myNodeOpt] (t1p) at (t1|-0,-2*\radius) {$x_0 x_1 z_1$ $y_{01} y_{11}$};
\draw[myEdgeOpt,->] (t1) -- (t1p);
\node[myNodeOpt] (tp) at (s0.west|-t0p) {$x_0 x_1$ $y_{00} y_{10}$ $y_{01} y_{11}$};
\draw[myEdgeOpt,->] (t0.north west) .. controls (t0.north west|-tp.south) .. (tp);
\draw[myEdgeOpt,->] (t1.south west) .. controls (t1p.south west|-t1p) and (t1p.south-|tp) .. (t1.south west-|tp.south west) -- (tp.south west);
\node at (0,-2.5*\radius) {$\cA(G_1)$ and $\cA(G_7)$};
}
\hfill
\tikz[baseline={(0,0)}]{
\node[myNodeSta] (s0) at (180:\radius) {$x_0$};
\node[myNodeSta] (s1) at (0:\radius) {$x_1$};
\node[myNode] (t) at (0,0) {$x_1 x_2$}
  edge[myEdge,<-] (s0)
  edge[myEdge,<-] (s1);
\node[myNode] (p00) at (s0|-0,0.5*\radius) {$x_0 y_{00}$}
  edge[myEdgeCyc,<-] (s0);
\node[myNode] (p01) at (s0|-0,-0.5*\radius) {$x_0 y_{01}$}
  edge[myEdgeCyc,<-] (s0);
\node[myNode] (p00p) at (s0|-0,\radius) {$x_0 x_1 y_{00}$}
  edge[myEdgeCyc,<-] (p00);
\draw[myEdge,->] (t) -- (p00p.343);
\node[myNode] (p01p) at (s0|-0,-\radius) {$x_0 x_1 y_{01}$}
  edge[myEdgeCyc,<-] (p01);
\draw[myEdge,->] (t) -- (p01p.17);
\node[myNode] (p10) at (s1|-0,0.66*\radius) {$x_1 y_{10}$}
  edge[myEdgeCyc,<-] (s1);
\node[myNode] (p11) at (s1|-0,-0.66*\radius) {$x_1 y_{11}$}
  edge[myEdgeCyc,<-] (s1); 
\node[myNodeTer] (t0) at (90:\radius) {$x_0 x_1$ $y_{00} y_{10}$}
  edge[myEdgeCyc,<-] (p00p)
  edge[myEdgeCyc,<-] (p10);
\node[myNodeTer] (t1) at (270:\radius) {$x_0 x_1$ $y_{01} y_{11}$}
  edge[myEdgeCyc,<-] (p01p)
  edge[myEdgeCyc,<-] (p11);
\node[myNode] (q00) at (-2*\radius,0|-p00) {$x_0 y_{00}$ $z_{00}$}
  edge[myEdge,<-] (p00);
\node[myNode] (q01) at (-2*\radius,0|-p01) {$x_0 y_{01}$ $z_{01}$}
  edge[myEdge,<-] (p01);
\node[myNode] (q10) at (s1|-0,1.33*\radius) {$x_1 y_{10}$ $z_{10}$}
  edge[myEdge,<-] (p10);
\node[myNode] (q11) at (s1|-0,-1.33*\radius) {$x_1 y_{11}$ $z_{11}$}
  edge[myEdge,<-] (p11);
\node[myNode] (qp) at (-2.5*\radius,0|-s0) {$x_0 x_1$ $y_{00} y_{01}$};
\draw[myEdge,->] (p00p) .. controls (p00p-|qp) .. (qp);
\draw[myEdge,->] (p01p) .. controls (p01p-|qp) .. (qp);
\node[myNodeOpt] (t0p) at (t0|-0,2*\radius) {$x_0 x_1 z_0$ $y_{00} y_{10}$};
\draw[myEdgeOpt,->] (t0) -- (t0p);
\node[myNodeOpt] (t1p) at (t1|-0,-2*\radius) {$x_0 x_1 z_1$ $y_{01} y_{11}$};
\draw[myEdgeOpt,->] (t1) -- (t1p);
\node[myNodeOpt] (qp0) at (qp|-t0p) {$x_0 x_1 y_{10}$ $y_{00} y_{01}$};
\node[myNodeOpt] (qp1) at (qp|-t1p) {$x_0 x_1 y_{11}$ $y_{00} y_{01}$};
\draw[myEdgeOpt,->] (qp.north west) -- (qp0.south west);
\draw[myEdgeOpt,->] (qp.south west) -- (qp1.north west);
\draw[myEdgeOpt,->] (t0.north west) .. controls (t0.north west|-0,1.5*\radius)  .. (-1.25*\radius,1.5*\radius) .. controls (qp0|-0,1.5*\radius) .. (qp0);
\draw[myEdgeOpt,->] (t1.south west) .. controls (t1.south west|-0,-1.5*\radius)  .. (-1.25*\radius,-1.5*\radius) .. controls (qp1|-0,-1.5*\radius) .. (qp1);
\node at (0,-2.5*\radius) {$\cA(G_2)$ and $\cA(G_6)$};
}
\caption{
The graphs $G_1, \ldots, G_7$.
The bottom left figure displays $\cA(G_7)$ and, without dashed nodes and arcs, it shows $\cA(G_1)$.
Analogously, the bottom right figure presents $\cA(G_6)$ or, without the dashed parts, $\cA(G_2)$.
Bold arcs emphasize the bad $2$-cycle, where starters are double framed and terminals bold framed.
}
\label{fig:counterexamples}
\end{figure}

This implies that $G_1, \ldots, G_7$ are forbidden induced subgraphs for $\cL$.
In the following section, we analyze the clique arrangement of these seven graphs and show that they share one particular cyclic property, related to bad $k$-cycles.

\section{Forbidden Cycles in Leaf Power Clique Arrangements}

As shown in \cite{NevRos2013}, strongly chordal graphs can be characterized by forbidden bad $k$-cycles in their clique arrangements, where $k \geq 3$.
But by Theorem \ref{thm_counterexamples}, this does not fully capture the cyclic structure that is forbidden in leaf powers.
In this section, we show that there are certain kinds of $2$-cycles which may not occur as a subgraph in the clique arrangement of a leaf power.
In particular, we call a $2$-cycle \emph{bad}, if for all $i,j \in \{0,1\}$ there is a directed path from starter $S_i$ to terminal $T_j$ that does not contain a node $X$ which fulfills $S_0 \cup S_1 \subseteq X \subseteq T_0 \cap T_1$.
The following theorem provides the main argument of this paper:
\begin{theorem}\label{thm_bad2cyc_counterexp}
Let $G=(V,E)$ be a strongly chordal graph with clique arrangement $\cA(G)=(\cX,\cE)$.
The graph $\cA(G)$ contains a bad $2$-cycle, if and only if $G$ contains one of the graphs $G_1, \ldots, G_7$ as an induced subgraph.
\end{theorem}
\begin{proof}
The proof starts by showing the first direction, that is, if $\cA(G)$ contains a bad $2$-cycle, then $G$ contains one of the graphs $G_1, \ldots, G_7$ as an induced subgraph.
Among the bad $2$-cycles of $\cA(G)$ we select a cycle with starters $S_0,S_1$ and terminals $T_0,T_1$ that primarily minimizes the summed cardinalities of the terminals $|T_0| + |T_1|$ and secondarily maximizes the summed cardinalities of the starters $|S_0| + |S_1|$.
Because $T_0$ and $T_1$ have a non-empty intersection, which contains at least $S_0 \cup S_1$, Lemma \ref{obs_node_intersection} provides a node $T = T_0 \cap T_1$.

In the following we provide a number of claims to support our arguments.
The proofs of all these claims are found in Section \ref{sec:techproofs}.
We start by shaping the bad $2$-cycle:
\begin{clai}\label{cla_pnodes}
For all $i,j \in \{0,1\}$ there is a path $B_{ij}$ from $S_i$ to $T_j$ that does not contain a node $X$ with $S_0 \cup S_1 \subseteq X \subseteq T_0 \cap T_1$, in particular $B_{ij}$ does not contain $T$, such that $B_{ij}$ contains a node $P_{ij}$ with
\begin{itemize}
\item[(1)]
$S_i \subseteq P_{ij} \subseteq T_j$, 
\item[(2)]
$S_{1-i} \not\subseteq P_{ij}$, 
\item[(3)]
$P_{ij} \not\subseteq T$, and
\item[(4)]
there exists a sink $Q_{ij}$ in $\cA(G)$ with $S_{1-i} \not\subseteq Q_{ij}$ that fulfills $P_{ij} = Q_{ij} \cap T_{j}$.
\end{itemize}
\end{clai}

In the following we refer to the nodes $P_{ij}$ by the $P$-nodes and we call $Q_{ij}$ the $Q$-nodes.
The pure existence of the $Q$-nodes does not directly imply that they are different:
\begin{clai}\label{cla_qnode:intersection}
For all $i,j,i',j' \in \{0,1\}$ with $(i,j) \not= (i',j')$, the sinks $Q_{ij}$ and $Q_{i'j'}$ differ.
\end{clai}
For the pairwise intersection between the $P$-nodes, Claim \ref{cla_pnodes} directly implies for all $i,j,j' \in \{0,1\}$ that $P_{ij} \not\subseteq P_{(1-i)j'}$.
We can now infer the following two additional statements about the intersections between the $P$-nodes and the intersections between the $Q$-nodes:
\begin{clai}\label{cla_pnode:intersection:vertical}
For all $i \in \{0,1\}$ it is true that $P_{i0} \cap P_{i1}=S_i$.
\end{clai}
\begin{clai}\label{cla_pqnode:intersection:hor:dia}
For all $i,i' \in \{0,1\}$ it is true that $P_{0i} \cap P_{1i'} \subseteq T$ and $Q_{0i} \cap Q_{1i'} \subseteq T$.
\end{clai}
We deduce that $P_{ij} \cap P_{i'j'} \subseteq T$ for all $i,j,i',j' \in \{0,1\}$ with $(i,j) \not= (i',j')$.
Following the construction of the $P$-nodes, we also know for all $i,j \in \{0,1\}$ that the set $P'_{ij} = P_{ij} \setminus T$ is not empty.

Using the collected facts about the mentioned nodes on the bad $2$-cycle, the next two claims start selecting vertices to construct one of the induced subgraphs $G_1, \ldots, G_7$:
\begin{clai}\label{cla_uvertices}
For all $i \in \{0,1\}$, the starter $S_i$ contains a vertex $u_i$ such that $u_i \not\in Q_{(1-i)0} \cup Q_{(1-i)1}$.
\end{clai}
\begin{clai}\label{cla_pwvertices}
For all $i,j \in \{0,1\}$, there is a vertex $w_{ij} \in Q_{ij} \setminus P_{ij}$ such that
\begin{itemize}
\item[(1)]
for all $i',j' \in \{0,1\}$ it is true that $w_{ij}=w_{i'j'} \Longleftrightarrow (i,j) = (i',j')$ and
\item[(2)]
$w_{ij}$ is neither adjacent to $u_{1-i},w_{i(1-j)},w_{(1-i)j)},w_{(1-i)(1-j)}$, nor to any vertex in $P'_{i(1-j)}$, in $P'_{(1-i)j}$ or in $P'_{(1-i)(1-j)}$.
\end{itemize}
\end{clai}

Depending on the edges between the six central vertices of $G_1, \ldots, G_7$, there exist up to two additional vertices in $G_4, \ldots, G_7$.
This dependency is also visible in the clique arrangement. 
Consider the sets $V_0 = P_{00} \cup P_{01}$, $V_1 = P_{10} \cup P_{11}$, $D_0 = P_{00} \cup P_{11}$ and $D_1 = P_{01} \cup P_{10}$ and moreover, for all $i,j \in \{0,1\}$ let $C_{ij} = V_i \cup D_j$.
If one of the sets $C_{ij}, i,j \in \{0,1\}$ induces a clique in $G$, then it follows that $T_0$ or $T_1$ are proper subsets of maximal cliques in $G$:
\begin{clai}\label{cla_tsinks}
For all $i,j \in \{0,1\}$ and $k = (i+j+1) \mod 2$, the node $T_k$ is not a sink in $\cA(G)$, if $C_{ij}$ is a clique in $G$.
\end{clai}
In such a case, if $C_{ij}$ is a clique, we select an additional vertex from the sink that is reachable from $T_k$:
\begin{clai}\label{cla_twvertices}
For all $i,j \in \{0,1\}$ and $k = (i+j+1) \mod 2$, if $C_{ij}$ is a clique in $G$, then there is a sink $T'_k$ which is reachable from $T_k$ and contains a vertex $w_k \in T'_k \setminus (P_{0k} \cup P_{1k} \cup T_{1-k})$ such that 
\begin{itemize}
\item[(1)]
$w_k$ is not one of the vertices $w_{1-k}, w_{00}, w_{01}, w_{10}, w_{11}$,
\item[(2)]
$w_k$ is neither adjacent to $w_{1-k}, w_{0(1-k)},w_{1(1-k)}$ nor to any vertex in $T_{1-k} \setminus T$, and
\item[(3)]
$w_k$ is adjacent to at most one vertex of $w_{0k}$ and $w_{1k}$.
\end{itemize}
\end{clai}

In the remainder of the proof we select the central vertices $v_{ij}$ from $P'_{ij}$ for all $i,j \in \{0,1\}$ to ultimately induce a forbidden subgraph.
But before explaining how to select these four vertices, we briefly summarize the results gathered in the proof so far.
By Claim \ref{cla_uvertices}, we know that there are vertices $u_0,u_1$ and, from the construction of the $P$-nodes in Claim \ref{cla_pnodes}, it follows that $\{u_0,u_1,v_{00},v_{10}\}$ and $\{u_0,u_1,v_{01},v_{11}\}$ are cliques in $G$, regardless of the choice of $v_{00},v_{01},v_{10},v_{11}$.
Moreover, by Claim \ref{cla_pwvertices}, there exists an independent set $\{w_{00},w_{01},w_{10},w_{11}\}$ in $G$ such that for all $i,j \in \{0,1\}$, the vertex $w_{ij}$ is adjacent to $u_i$ and $v_{ij}$ but not to any of the vertices $u_{1-i},v_{i(1-j)},v_{(1-i)j},v_{(1-i)(1-j)}$. 
Finally, Claim \ref{cla_twvertices} states that certain circumstances imply the existence of two non-adjacent vertices $w_0$ and $w_1$ in $G$ that are both adjacent to $u_0$ and $u_1$ and such that for all $k \in \{0,1\}$ it is true that $w_k$ is adjacent to $v_{0k}$ and $v_{1k}$ but not adjacent to $v_{0(1-k)},v_{1(1-k)},w_{0(1-k)}$ and $w_{0(1-k)}$.
The claim leaves it open, if $w_k$ can be adjacent to either $w_{0k}$ or $w_{1k}$ and, consequently, we cope with this problem during the following vertex selection.
These facts are subsequently used without explicit mentioning.

Moreover, in the following vertex selection we write
\[G_i(x_0,x_1,y_{00},y_{01},y_{10},y_{11},z_{00},z_{01},z_{10},z_{11},[z_0,z_1])\]
to state that $G$ contains an induced $G_i$ for $i \in \{1, \ldots, 7\}$ on vertices $x_0$, $x_1$, $y_{00}$, $y_{01}$, $y_{10}$, $y_{11}$, $z_{00}$, $z_{01}$, $z_{10}$, $z_{11}$, optionally including $z_0,z_1$.
Hence, both vertex sets $x_0,x_1,y_{00},y_{10}$ and $x_0,x_1,y_{01},y_{11}$ form a clique in $G$ and every $z_{ij}$ is exactly adjacent to $x_i,y_{ij}$.
Depending on $i$ the vertices $z_0$ and $z_1$ are present and $z_i$ is exactly adjacent to $x_0,x_1,y_{0i},y_{1i}$ for $i \in \{0,1\}$.
The adjacency between $y_{00},y_{01},y_{10},y_{11}$ depends on $i$, too.

To find suitable vertices for the forbidden induced subgraphs, we have to distinguish between three cases:
\begin{enumerate}
\item\textbf{Assume that at most one of the sets $V_0, V_1, D_0, D_1$ is a clique in $G$:}
Because of symmetry we just have the following two subcases:
\begin{enumerate}
\item\textbf{Assume that at most $V_0$ is a clique in $G$:}
Because $D_0,D_1$ are not cliques, we can select vertices $v_{ij} \in P'_{ij}$ for all $i,j \in \{0,1\}$ such that $v_{00} \nad v_{11}$ and $v_{01} \nad v_{10}$.
At most one of the edges $v_{00}v_{01}$ or $v_{10}v_{11}$ is present in $E$ because otherwise $v_{00} \ad v_{01} \ad v_{11} \ad v_{10} \ad v_{00}$ is an induced $C_4$ in $G$.
If $v_{00} \nad v_{01}$ and $v_{10} \nad v_{11}$, then
\[G_1(u_0,u_1,v_{00},v_{01},v_{10},v_{11},w_{00},w_{01},w_{10},w_{11}).\]
Clearly, if $V_0$ is a clique, then $v_{00} \ad v_{01}$ and $v_{10} \nad v_{11}$ and then we have
\[G_2(u_0,u_1,v_{00},v_{01},v_{10},v_{11},w_{00},w_{01},w_{10},w_{11}).\]
\item\textbf{Assume that at most $D_0$ is a clique in $G$:}
Analogously to the previous case, we can select $v_{ij} \in P'_{ij}$ for all $i,j \in \{0,1\}$ such that $v_{00} \nad v_{01}$ and $v_{10} \nad v_{11}$ and again, either $v_{00} \ad v_{11}$ or $v_{01} \ad v_{10}$ as otherwise $v_{00} \ad v_{11} \ad v_{01} \ad v_{10} \ad v_{00}$ is an induced $C_4$ in $G$.
The case of $v_{00} \nad v_{01}$ and $v_{10} \nad v_{11}$ yields an induced $G_1$ and has already been handled in the first case.
If without loss of generality $v_{00} \ad v_{11}$, then
\[G_3(u_0,u_1,v_{00},v_{01},v_{10},v_{11},w_{00},w_{01},w_{10},w_{11}).\]
\end{enumerate}
\item\textbf{Assume that exactly two of the sets $V_0,V_1,D_0,D_1$ are cliques in $G$:}
If $V_0$ and $V_1$ are cliques but not $D_0$ and $D_1$, then vertices $v_{ij} \in P'_{ij}$ exist for all $i,j \in \{0,1\}$ such that $v_{00} \nad v_{11}$ and $v_{01} \nad v_{10}$ and consequently, $v_{00} \ad v_{01} \ad v_{11} \ad v_{10} \ad v_{00}$ is an induced $C_4$.
Analogously, $D_0$ and $D_1$ being the cliques implies $v_{00} \ad v_{11} \ad v_{01} \ad v_{10} \ad v_{00}$ as an induced $C_4$.
Because of this and symmetry, we have only one remaining case, namely $V_0$ and $D_0$ are the cliques and this implies that $C_{00}$ is a clique.

Next we show that there exist vertices $v_{ij} \in P'_{ij}$ for all $i,j \in \{0,1\}$ such that $v_{01} \nad v_{10}$ and $v_{10} \nad v_{11}$.
For that purpose, assume that every vertex in $P'_{10}$, that is adjacent to some vertex in $P'_{(1-k)1}$ for $k \in \{0,1\}$, is also adjacent to all vertices in $P'_{k1}$.
Then, as $V_1$ and $D_1$ are not cliques, there are vertices $x\not=y \in P'_{10}$ such that there is $x' \in P'_{01}$ and $y' \in P'_{11}$ with $x \nad x'$ and $y \nad y'$.
By our assumption, it follows that $x \ad y'$ and $x' \ad y$ and hence, there is $x \ad y \ad x' \ad y' \ad x$, an induced $C_4$ in $G$.
Consequently, the assumption was wrong and we can select the vertices such that $v_{01} \nad v_{10}$ and $v_{10} \nad v_{11}$.

Because $C_{00}$ is a clique, it follows by Claim \ref{cla_twvertices} that $w_1$ exists, and if $w_1$ is neither adjacent to $w_{01}$ nor to $w_{11}$, then
\[G_4(u_0,u_1,v_{00},v_{01},v_{10},v_{11},w_{00},w_{01},w_{10},w_{11},w_1).\]
Otherwise, if $w_1$ is adjacent to $w_{10}$, then we get
\[G_3(u_0,u_1,v_{00},w_1,v_{10},v_{11},w_{00},w_{01},w_{10},w_{11}),\]
and if $w_1 \ad w_{11}$, then
\[G_2(u_0,u_1,v_{00},v_{01},v_{10},w_1,w_{00},w_{01},w_{10},w_{11}).\]
\item\textbf{Assume that at least three of the sets $V_0,V_1,D_0,D_1$ are cliques in $G$:}
In this case, we select any vertex $v_{ij} \in P'_{ij}$ for all $i,j \in \{0,1\}$.
As at least one of the sets $C_{00} = V_0 \cup D_0$ or $C_{11} = V_1 \cup D_1$ is a clique, it follows from Claim \ref{cla_twvertices} that $w_1$ exists.
Analogously, $C_{01} = V_0 \cup D_1$ or $C_{10} = V_1 \cup D_0$ is a clique and thus, $w_0$ exists.

Assume first that $w_0$ and $w_1$ are completely disjoint from $w_{00}, w_{01}, w_{10}, w_{11}$.
By symmetry we just have to consider the cases of (i) $v_{01} \nad v_{10}$, which leads to
\[G_5(u_0,u_1,v_{00},v_{01},v_{10},v_{11},w_{00},w_{01},w_{10},w_{11},w_0,w_1),\]
(ii) $v_{10} \nad v_{11}$, which yields
\[G_6(u_0,u_1,v_{00},v_{01},v_{10},v_{11},w_{00},w_{01},w_{10},w_{11},w_0,w_1),\]
and (iii) $v_{00},v_{01},v_{10},v_{11}$ are a clique where
\[G_7(u_0,u_1,v_{00},v_{01},v_{10},v_{11},w_{00},w_{01},w_{10},w_{11},w_0,w_1).\]

Finally, we have to check all the cases where $w_0$ or $w_1$ are adjacent to one of the vertices $w_{00}, w_{01}, w_{10}, w_{11}$.
Because of symmetry we can simply assume that $w_0$ is adjacent to $w_{10}$.

Assume that $w_1$ is neither adjacent to $w_{01}$ nor to $w_{11}$.
Then (iv) $v_{00} \ad v_{01}$ and $v_{00} \nad v_{11}$ implies 
\[G_2(u_0,u_1,v_{00},v_{01},w_0,v_{11},w_{00},w_{01},w_{10},w_{11}),\]
(v) $v_{00} \nad v_{01}$ and $v_{00} \ad v_{11}$ yields
\[G_3(u_0,u_1,v_{00},v_{01},w_0,v_{11},w_{00},w_{01},w_{10},w_{11}),\]
and (vi) $v_{00} \ad v_{01}$ and $v_{00} \ad v_{11}$ gives
\[G_4(u_0,u_1,v_{00},v_{01},w_0,v_{11},w_{00},w_{01},w_{10},w_{11},w_1).\]

If $w_1 \ad w_{11}$, then (vii) $v_{00} \nad v_{01}$ implies
\[G_1(u_0,u_1,v_{00},v_{01},w_0,w_1,w_{00},w_{01},w_{10},w_{11}),\]
and (viii) $v_{00} \ad v_{01}$ yields
\[G_2(u_0,u_1,v_{00},v_{01},w_0,w_1,w_{00},w_{01},w_{10},w_{11}).\]
Moreover, if $w_1 \ad w_{01}$, then (ix) $v_{00} \nad v_{11}$ results in
\[G_1(u_0,u_1,v_{00},w_1,w_0,v_{11},w_{00},w_{01},w_{10},w_{11}),\]
and (x) $v_{00} \ad v_{11}$ provides
\[G_3(u_0,u_1,v_{00},w_1,w_0,v_{11},w_{00},w_{01},w_{10},w_{11}).\]
\end{enumerate}

The following shows the converse direction, that is, if $G$ contains one of $G_1, \ldots, G_7$ as an induced subgraph, then $\cA(G)$ has a bad $2$-cycle. 

We basically use Lemma \ref{lma_cliqueArrangement_inducedSubgraphs}.
The clique arrangement of all graphs $G_1, \ldots, G_7$ contains a bad $2$-cycle with starters $S_0=\{x_0\}$, $S_1=\{x_1\}$ and terminals $T_0=\{x_0,x_1,y_{00},y_{10}\}$, $T_1=\{x_0,x_1,y_{01},y_{11}\}$.
Moreover, there are nodes $P_{ij} = \{x_i,y_{ij}\}$, $Q_{ij} = \{x_i,y_{ij},z_{ij}\}$ such that $S_i \dad \ldots \dad P_{ij} \dad \ldots \dad T_j$ and $P_{ij} \dad \ldots \dad Q_{ij}$ for all $i,j \in \{0,1\}$.
If $G$ contains an induced subgraph $G_1, \ldots, G_7$, then there is a function $\phi$, that maps these nodes to some nodes of the clique arrangement $\cA(G)$ such that $\phi(S_i) \dad \ldots \dad \phi(P_{ij}) \dad \ldots \dad \phi(T_j)$ and $\phi(P_{ij}) \dad \ldots \dad \phi(Q_{ij})$ for all $i,j \in \{0,1\}$.

Assume that at least one of these four paths in $\cA(G)$, say $\phi(S_0) \dad \ldots \dad \phi(T_0)$, contains a node $X$ with $\phi(S_0) \cup \phi(S_1) \subseteq X \subseteq \phi(T_0) \cap \phi(T_1)$.
If $X$ is situated on the subpath $\phi(S_0) \dad \ldots \dad \phi(P_{00})$, then it follows that $X \subset Q_{00}$ and, hence, $x_1 \ad z_{00}$, a contradiction.

Hence, $X$ is on the subpath $\phi(P_{00}) \dad \ldots \dad \phi(T_0)$.
Here, $\phi(P_{00})$ is a subset of $X \subseteq \phi(T_0) \cap \phi(T_1)$ and thus, also a subset of $\phi(T_1)$.
This means that $y_{00} \in \phi(T_1)$, which implies $y_{00} \ad y_{01}$ and $y_{00} \ad y_{11}$.
Consequently, we are in the case were the induced subgraph in $G$ is one of $G_4, \ldots, G_7$.
The clique arrangement of all these graphs contains a sink $T'_1 = \{x_0,x_1,y_{01},y_{11},z_1\}$ that is reached from $T_1$.
In $\cA(G)$, we have $\phi(T_1) \dad \ldots \dad \phi(T'_1)$, thus, $\phi(P_{00}) \subset \phi(T'_1)$, which finally means that $y_{00} \ad z_1$, a contradiction.

Hence, $X$ does not exist and $\cA(G)$ contains a bad $2$-cycle with starters $\phi(S_0), \phi(S_1)$ and terminals $\phi(T_0), \phi(T_1)$.
\end{proof}

The main theorem, presented in this section, and Theorem \ref{thm_counterexamples} lead to the following conclusion:
\begin{corollary}\label{cor_bad2cyc_leafp}
Let $G=(V,E)$ be a graph in $\cL$ with clique arrangement $\cA(G)=(\cX,\cE)$.
The graph $\cA(G)$ does not contain a bad $2$-cycle.
\end{corollary}
Hence, leaf powers fit naturally into the hierarchy of chordal graphs, right between strongly chordal graphs, which have clique arrangements without bad $k$-cycles for $k \geq 3$, and ptolemaic graphs, whose clique arrangements are entirely free of cycles.

\section{Conclusion and Future Directions}

In this paper, we were able to indicate that leaf powers $\cL$ are a natural subclass of strongly chordal graphs by showing that their clique arrangements are not only free of bad $k$-cycles for $k \geq 3$ but also for $k=2$. 
Moreover, we proved that the clique arrangement of a strongly chordal graph $G$ comprises a bad $2$-cycle, if and only if $G$ contains at least one of $G_1, \ldots, G_7$ as an induced subgraph.
This means that, beside the forbidden induced subgraphs of strongly chordal graphs, that is, the family of suns, this finite number of graphs describe a cyclic composition of cliques that is not realizable by a $k$-leaf root for any $k \geq 2$.

It remains for future work to find a complete characterization of $\cL$ in terms of forbidden subgraphs.
During our deep analysis of leaf powers we have considered a huge variety of graphs and their clique arrangements.
We have not a single example of a graph $G$ that has a clique arrangement $\cA(G)$ without bad $k$-cycles for $k \geq 2$, where a corresponding leaf root of $G$ is unknown.
Therefore, we conjecture that a strongly chordal graph $G$ has a $k$-leaf root for some $k \geq 2$, if and only if $\cA(G)$ is free of bad $2$-cycles.
If this was true, a polynomial time recognition algorithm is straight found by the efficient recognition of strongly chordal graphs and the possibility to check for a finite number of induced subgraphs in polynomial time.

Answering this question implies the challenge of constructing leaf roots from bad-cycle-free clique arrangements.
This turns out to be sophisticated, especially if the clique arrangement has $2$-cycles that are not bad.
%These nodes seem to break the symmetry of a cycle, which should be used in the construction.
%We leave this as an open problem.

\section{Technical Proofs}
\label{sec:techproofs}

\subsection*{The Proof of Lemma \ref{lma_node_to_sink_intersection}}

\begin{proof}
First of all, we know that $Y \not\subseteq Z$ and $Z \not\subseteq Y$ as otherwise $X,Y,Z$ are not distinct nodes.
Let $\cX_Y$ be the set of sinks reachable from $Y$ and $\cX_Z$ be the set of sinks reachable from $Z$.
If $Y$ and $Z$ are sinks themselves then $\cX_Y=\{Y\}$ and $\cX_Z=\{Z\}$ and trivially we set $C_1=Y$ and $C_2=Z$.

If just one of the nodes, $Y$ or $Z$, is a sink, without loss of generality, $\cX_Y=\{Y\}$, then Lemma \ref{lma_stronglyChordal_cliqueCutStructure} implies that there are sinks $A,B \in \cX_Z$ such that $A \not= B$ and $Z=A \cap B$.
Clearly, $A \not= Y$ and $B \not= Y$ as otherwise $Z \subseteq Y$.
By Lemma \ref{obs_node_intersection}, $A' = A \cap Y$ and $B' = B \cap Y$ are nodes in $\cX$ because both contain $X$.
It is straight forward that $A'\not=B'$ and moreover, $A'\not=Z$ and $B'\not=Z$ as otherwise $Z \subseteq Y$.
We have a $3$-cycle with starters $A',B',Z$ and terminals $A,B,Y$.
By Theorem \ref{thm_badcycles}, the cycle is not bad and as $Z \not\subseteq Y$ it follows that $A' \subseteq B$ or $B' \subseteq A$.
If $A' \subseteq B$, then $A' \subseteq A$ and $Z= A \cap B$ imply that $A' \subseteq Z$.
From $A' \subseteq Y$ and $X=Y \cap Z$ we obtain $A' \subseteq X$ and, as $X \subseteq A'$, we have $X=A'$.
Hence, in this case $X = A \cap Z$ and we set $C_1=A$ and $C_2=Y$.
In an analogous fashion we get $X = B \cap Y$, if $B' \subseteq A$ and then we set $C_1=B$ and $C_2=Y$.

If none of the nodes $Y,Z$ is a sink, then Lemma \ref{lma_stronglyChordal_cliqueCutStructure} implies that there are sinks $A\not= B \in \cX_Y$ and $C \not= D \in \cX_Z$ such that $Y=A \cap B$ and $Z=C \cap D$.
If $Y$ reaches one of the sinks $C$ and $D$ or $Z$ reaches one of the sinks $A$ and $B$, without loss of generality, $Z$ reaches $B$, then $Z \subseteq B \cap C$.
Notice that $Z$ cannot reach $A$ in this case, as otherwise $Z \subseteq Y$.
By Lemma \ref{obs_sink_intersection}, $Z$ is exactly the intersection of all cliques in $\cX_Z,$ which includes $B$.
By definition, $Z = C \cap D$ and, hence, we also have $Z = B \cap C \cap D$.
As $B \cap C \cap D \subseteq B \cap C$, we conclude that $Z = B \cap C$.
The node $A' = A \cap C$ exists by Lemma \ref{obs_node_intersection} because it contains $X$ as a subset.
Clearly, if $Y=A'$ or $Z = A'$ then $A' = A \cap B \cap C$.
Otherwise, we have a $3$-cycle with starters $Y,Z,A'$ and terminals $A,B,C$.
By Theorem \ref{thm_badcycles}, the cycle is not bad and as $Y \subseteq C$ or $Z \subseteq A$ implies $Y \subseteq Z$ or $Z \subseteq Y$, we have $A' \subseteq A \cap B \cap C$, again.
Because $Y = A \cap B$ and $Z = B \cap C$, this implies that $A' \subseteq Y$ and $A' \subseteq Z$, and because $X = Y \cap Z$, we have $A' \subseteq X$ and thus, $X=A'$.
Hence, $X = A \cap C$ and we set $C_1=A$ and $C_2=C$.

If neither $Y$ reaches the sinks $C$ or $D$ nor $Z$ reaches the sinks $A$ or $B$, then we consider the nodes $A' = A \cap C$ and $B' = B \cap D$, which exist by Lemma \ref{obs_node_intersection}, as they all contain $X$ as a subset.
Clearly, if $A'$ or $B'$ coincides with $Y$, then $Y$ reaches one of the sinks $C$ or $D$ and analogously, $A'$ and $B'$ are not $Z$.
Moreover, $A'=B'$ implies that $A'$ is a subset of $Y$ and $Z$ and, by that, a subset of $X$, which means that $X=A'= A \cap C$ and that $C_1=A$ and $C_2=C$.
Otherwise, as $Y\not=Z$, we get a $4$-cycle with starters $A',B',Y,Z$ and terminals $A,B,C,D$.
Theorem \ref{thm_badcycles} states that the cycle is not bad and as $Y \not\subseteq C \cup D$ and $Z \not\subseteq A \cup B$, it must be true, without loss of generality, that $A' \subseteq B$ and hence, we get that $A' \subseteq Y$ and that there is a $3$-cycle with starters $A',B',Z$ and terminals $B,C,D$.
This cycle is not bad, either, and hence, either $A' \subseteq D$ or $B' \subseteq C$.
If $A' \subseteq D$, then $A' \subseteq Z$, which implies $A' \subseteq X$ and thus, $X=A'=A \cap C$ and then $C_1=A$ and $C_2=C$.

The case $B' \subseteq C$ implies that $B' \subseteq Z$ and then we consider the node $C' = A \cap D$, which exists because it contains $X$ as a subset.
If $B'=C'$, $B'$ is contained in $X = Y \cap Z$, which means that $X=B' = B \cap D$ and that $C_1=B$ and $C_2=D$.
Otherwise, we have a $3$-cycle with starters $B',C',Y$ and terminals $A,B,D$.
Because the cycle is not bad, it is true that $B' \subseteq A$ or $C' \subseteq B$.
If $B' \subseteq A$, then $B' \subseteq Y$, which implies $B' \subseteq X$ and thus, $X=B'=B \cap D$ and then $C_1=B$ and $C_2=D$.
In the other case, $C'$ is contained in $A$ and $B$ and thus, in $Y$.
Moreover, $C'$ is a subset of $B \cap D$, thus, a subset of $B'$, and, consequently, contained in $C \cap D$, which means that $C'$ is also in $Z$.
Together, this implies that $X=C'= A \cap D$ and that $C_1=A$ and $C_2=D$.
\end{proof}

\subsection*{The Proof of Lemma \ref{lma_cliqueArrangement_inducedSubgraphs}}

\begin{proof}
We first fix a function $\phi$.
For that purpose notice that for every maximal clique $C$ of $H$, there is at least one maximal clique $C'$ in $G$ such that $C \subseteq C'$ and we define $\phi(C) = C'$.
Notice that $C = \phi(X) \cap V$.
Moreover, for every node $X \in \cX$ that is not a maximal clique, there is the subset $C_1, \ldots, C_k$ of maximal cliques in $H$ that are reached in $\cA(H)$ by a directed path from $X$.
As $X = C_1 \cap \ldots \cap C_k$, we define $\phi(X) = X'$ for the node $X' \in \cX'$ that fulfills  $X' = \phi(C_1) \cap \ldots \cap \phi(C_k)$.

The proof is completed by showing the two declared properties for all $X,Y \in \cX$:
\begin{enumerate}
\item
Since $\phi$ is a function, $X=Y$ implies $\phi(X) = \phi(Y)$.
Conversely, if $\phi(X) = \phi(Y)$ but $X\not=Y$, then there are non-adjacent vertices $x \in X \setminus Y$ and $y \in Y \setminus X$, which are, by definition, both in $\phi(X)$, a contradiction.
\item
By definition, there is a directed path from $\phi(X)$ to $\phi(Y)$ in $\cA(G')$, if $\phi(X) \subseteq \phi(Y)$.
As $X = \phi(X) \cap V$ and $Y = \phi(Y) \cap V$ this implies $X \subseteq Y$, which, by definition, means that there is a directed path from $X$ to $Y$ in $\cA(G)$.

Conversely, let there be a directed path from $X$ to $Y$ in $\cA(G)$, thus, let $X \subseteq Y$.
This means that in $\cA(G)$ the set $C_1, \ldots, C_k$ of maximal cliques reached from $Y$ is a subset of the maximal cliques $C_1, \ldots, C_\ell$ reached from $X$, hence, $k \leq \ell$.
Consequently, $\phi(C_1) \cap \ldots \cap \phi(C_\ell) = \phi(X) \subseteq \phi(Y) = \phi(C_1) \cap \ldots \cap \phi(C_k)$, which implies that there is a directed path from $\phi(X)$ to $\phi(Y)$ in $\cA(G')$.
\end{enumerate}
\end{proof}

\subsection*{The Proof of Theorem \ref{thm_counterexamples}}

\begin{proof}
The proof works basically the same as in \cite{BroLow1986}.
Assume that at least one of the graphs $G_1, \ldots, G_7$ is a $k$-leaf power of a tree $T$ for some $k \geq 2$ and that $x'_0, x'_1,y'_{00},y'_{01},y'_{10},y'_{11}$ are the parent nodes of the leaves $x_0, x_1,y_{00},y_{01},y_{10},y_{11}$ in $T$.

Consider for all $i,j \in \{0,1\}$ the path $P_{ij} = x'_i \ad \ldots \ad y'_{ij}$ in $T$ as well as, for all $i \in \{0,1\}$, the path $P_i = y'_{0i} \ad \ldots \ad y'_{1i}$ in $T$.
From Lemma \ref{lma_criticalEdges} we get that $P_{00} \cap P_{10} = \emptyset$ and $P_{00} \cap P_{11} = \emptyset$.
Similarly, Lemma \ref{lma_criticalEdges} implies that $P_{01} \cap P_{10} = \emptyset$ and $P_{01} \cap P_{11} = \emptyset$.
This means that the subtree $T_0$ of $T$ given by the union $P_{00} \cup P_{01}$ is disjoint from the subtree $T_1$ of $T$ given by $P_{10} \cup P_{11}$.

As $T$ is a tree, there is a node $z$ situated on the path connecting the subtrees $T_0$ and $T_1$ such that $z$ is on every path $x \ad \ldots \ad y$ in $T$ that connects a node $x$ from $T_0$ and a node $y$ from $T_1$.
In particular, that also means that $z$ is on $P_0$, if $x=y'_{00}$ and $y=y'_{10}$, and that $z$ is on $P_1$, if $x=y'_{01}$ and $y=y'_{11}$.
Hence, $P_0 \cap P_1 \not= \emptyset$, as both paths contain $z$, which is a contradiction to Lemma \ref{lma_criticalEdges}.
\end{proof}

\subsection*{The Proofs of Claims in Theorem \ref{thm_bad2cyc_counterexp}}

The claims proved in the following are stated in a general and simple fashion, and they often use indices $i,j \in \{0,1\}$ for the occurring nodes.
However, because the bad $2$-cycle is symmetric, the proofs always show the individual statements just for the case $i=j=0$ without explicit indication.

\subsubsection*{The Proof of Claim \ref{cla_pnodes}}

\begin{proof}
As mentioned, we show the claim only for $i=j=0$.

We start by choosing an arbitrary path $B_{00}$ from $S_0$ to $T_0$ that does not contain a node $X$ with $S_0 \cup S_1 \subseteq X \subseteq T_0 \cap T_1$, which exists by the definition of bad $2$-cycles.
Obviously, this implies that the node $T$ is not on $B_{00}$.

Firstly, there are nodes $P,P'$ on the path $B_{00} = S_0 \dad \ldots \dad P \dad P' \dad \ldots \dad T_0$ that are joined by an arc $P \dad P'$ such that $P \subseteq T$ and $P' \not\subseteq T$ and $P' \not= T_0$, hence, on $B_{00}$, the node $P$ is the last exit to $T$.
Clearly, we have $S_0 \subseteq T$ and thus, if such arc does not exist, then every node on the path, except $T_0$ itself, would be a subset of $T$.
Because $T$ is not on $B_{00} = S_0 \dad \ldots \dad Q \dad T_0$, even the predecessor $Q$ of $T_0$ reaches $T$ by a directed path.
Hence, as $T \subset T_0$, there is a directed path $Q \dad \ldots \dad T \dad \ldots \dad T_0$ and, consequently, the arc $Q \dad T_0$ is transitive, a contradiction.

Next we show that $S_1 \subseteq P'$ implies also that $S_1 \subseteq P$.
This can be seen by the use of the intersection node $X = P' \cap T$, which entirely contains $S_1$ because $S_1 \subset P'$ and $S_1 \subseteq T$.
As $P' \not\subseteq T$ and $X \subseteq T$, it follows that $X$ is not equal to the node $P'$.
Moreover, since $P \subseteq P'$ and $P \subseteq T$, it follows that $P \subseteq X$ and hence, there is a path $P \dad \ldots \dad X \dad \ldots \dad T$.
But $X$ cannot be a node on that path, unless $X=P$, because otherwise $P \dad P'$ would be a transitive arc.
But $X=P$ implies that $B_{00} = S_0 \dad \ldots \dad X=P \dad \ldots \dad T_0$ passes a node that fulfills $S_0 \cup S_1 \subseteq X = P \subseteq T$, which is a contradiction to the selection of the bad $2$-cycle.
Hence, $S_1 \not\subseteq P'$ must be true.

However, $P'$ is not necessarily the node $P_{00}$ we are looking for.
Particularly, it may happen that no sink $Q$ of $\cA(G)$ fulfills $Q \cap T_0 = P'$.
For that reason, let $Q_1, \ldots, Q_r$ be the sinks reachable from $P'$ by directed paths and let $P'_1 = Q_1 \cap T_0, \ldots, P'_r = Q_r \cap T_0$.
Because $P' = P'_1 \cap \ldots \cap P'_r$, Lemma \ref{obs_sink_intersection} implies that, if $S_1 \subseteq P'_i$ for all $i \in \{1, \ldots, r\}$, then $S_1 \subseteq P'$.
Hence, we can select $i \in \{1, \ldots,r\}$ such that $S_1 \not\subseteq P'_i$ and we set $P_{00} = P'_i$ and $Q_{00} = Q_i$.

Of course, it may happen that $P_{00}$ is not on the path $B_{00}$, but now we have a new path $B' = S_0 \dad \ldots \dad P' \dad \ldots \dad P_{00} \dad \ldots \dad T_0$.
We use $B'$ as a replacement for $B_{00}$, because it is easy to see that it does not contain a node $X$ with $S_0 \cup S_1 \subseteq X \subseteq T$, too.
If such a node $X$ was on the subpath $S_0 \dad \ldots \dad P_{00}$, then $S_1 \subset P_{00}$, and, if it was on the subpath $P_{00} \dad \ldots \dad T_0$, then $P_{00} \subset T$, which both contradicts the construction of $P_{00}$.

Finally, as $P_{00} = Q_{00} \cap T_0$, it follows that $S_1 \not\subseteq Q_{00}$, as otherwise $S_1 \subseteq T_0$ implies that $S_1 \subseteq P_{00}$, too.
\end{proof}

\subsubsection*{The Proof of Claim \ref{cla_qnode:intersection}}

\begin{proof}
The case $Q_{00} = Q_{1j'}$ is impossible for all $j' \in \{0,1\}$, because then $S_1 \subseteq Q_{1j'}$ implies $S_1 \subseteq Q_{00}$, which is forbidden by Claim \ref{cla_pnodes}.
If we assume that $Q_{00} = Q_{01}$, then we get a $3$-cycle with starters $P_{00},P_{01},S_1$ and terminals $T_0,T_1,Q_{00}$.
Certainly, $P_{00}$ is not contained in $T_1$ as otherwise $P_{00} \subseteq T_0$ implies $P_{00} \subseteq T$, which is forbidden by Claim \ref{cla_pnodes}.
Similarly, we get that $P_{01} \not\subseteq T_0$.
That $S_1 \not\subseteq Q_{00}$ is a direct consequence of Claim \ref{cla_pnodes}.
Hence, the $3$-cycle is bad, a contradiction to Theorem \ref{thm_badcycles}.
\end{proof}

\subsubsection*{The Proof of Claim \ref{cla_pnode:intersection:vertical}}

\begin{proof}
Let $S$ be the node representing the intersection $P_{00} \cap P_{01}$, which exists by Lemma \ref{obs_node_intersection} as $S_0 \subseteq S$.
If we assume that $S_0 \not= S$, then we have two paths $B'_{00} = S \dad \ldots \dad P_{00} \dad \ldots \dad T_0$ and $B'_{01} = S \dad \ldots \dad P_{01} \dad \ldots \dad T_1$ and we obtain a $2$-cycle with starters $S,S_1$ and terminals $T_0,T_1$.
We show that this cycle is bad by arguing that none of $B'_{00}$, $B'_{01}$, $B_{10}$, and $B_{11}$ contains a node $X$ that fulfills $S \cup S_1 \subseteq X \subseteq T$.
Clearly, the existence of $X$ on one of $B_{10}$, $B_{11}$, $P_{00} \dad \dots \dad T_0$, and $P_{01} \dad \dots \dad T_1$ contradicts to the choice of $B_{00}$, $B_{01}$, $B_{10}$, and $B_{11}$.
%Because $S_0 \subset S$, it follows that the primal bad $2$-cycle had a node that fulfills $S_0 \cup S_1 \subseteq X \subseteq T$.

If $X$ was on the path $S \dad \ldots \dad P_{00}$, then we would get $S_1 \subset P_{00}$, which has been eliminated in Claim \ref{cla_pnodes}.
Similarly, the path $S \dad \ldots \dad P_{01}$ does not contain $X$, and hence, the new $2$-cycle is bad.
But this contradicts to the choice of the primal bad $2$-cycle, because, by $S_0 \subset S$, we obtain $|S_0|+|S_1| < |S|+|S_1|$.
\end{proof}

\subsubsection*{The Proof of Claim \ref{cla_pqnode:intersection:hor:dia}}

\begin{proof}
If $Q_{00} \cap Q_{1i'} = \emptyset$, then clearly $Q_{00} \cap Q_{1i'} \subseteq T$.
Otherwise, let $Q$ be the intersection node for $Q_{00} \cap Q_{1i'}$, which exists by Lemma \ref{obs_node_intersection}.
We get a $3$-cycle with starters $S_0,S_1,Q$ and terminals $Q_{00},Q_{10},T$.
If $Q \not\subseteq T$ then the $3$-cycle is bad, because $S_0 \not\subseteq Q_{10}$ and $S_1 \not\subseteq Q_{00}$ by Claim \ref{cla_pnodes}.
This contradicts Theorem \ref{thm_badcycles}.

Clearly, we have $P_{00} \cap P_{10} \subseteq Q_{00} \cap Q_{10} \subseteq T$ and $P_{00} \cap P_{11} \subseteq Q_{00} \cap Q_{11} \subseteq T$.
\end{proof}

\subsubsection*{The Proof of Claim \ref{cla_uvertices}}

\begin{proof}
If $u_0$ does not exist, then $S_0$ is a subset of $Q_{10} \cup Q_{11}$.
As $S_0$ cannot be entirely contained in a single set, $Q_{10}$ or $Q_{11}$, we find two distinct nodes $X = S_0 \cap Q_{10}$ and $Y = S_0 \cap Q_{11}$ by Lemma \ref{obs_node_intersection}.
The same lemma reveals the existence of a node $Z = Q_{10} \cap Q_{11}$, because $Z$ contains at least as $S_1$.

We get a $3$-cycle with starters $X,Y,Z$ and terminals $S_0,Q_{10},Q_{11}$.
We show that this cycle is bad by the help of $S_0 = (S_0 \cap Q_{10}) \cup (S_0 \cap Q_{11})$.
Firstly, $X$ cannot be a subset of $Q_{11}$, because otherwise $Q_{11}$, which already contains $Y=S_0 \cap Q_{11}$, contains also $S_0 \cap Q_{10}$, which would imply that $S_0 \subseteq Q_{11}$.
Secondly and similarly, $Y$ cannot be a subset of $Q_{10}$, because otherwise $S_0 \subseteq Q_{10}$.
Finally, by $S_1 \subseteq Z$, it follows that $Z \not\subseteq S_0$.
This bad $3$-cycle contradicts Theorem \ref{thm_badcycles}, hence, the node $u_0$ exists.
\end{proof}

\subsubsection*{The Proof of Claim \ref{cla_pwvertices}}

\begin{proof}
As the $Q$-nodes represent distinct maximal cliques in $G$, Lemma \ref{lma_vertex_node_none_adjacency} allows to select vertices $x \in Q_{00} \setminus Q_{10}$ and  $y \in Q_{00} \setminus Q_{11}$ such that $x$ is not adjacent to any vertex in $Q_{10} \setminus Q_{00}$ and $y$ is not adjacent to any vertex in $Q_{11} \setminus Q_{00}$.

We show that at least one of $x$ and $y$ is not adjacent to all vertices in $(Q_{10} \cup Q_{11}) \setminus Q_{00}$.
If $x=y$ we are done.
Otherwise, assume that $x$ has a neighbor $x' \in Q_{11} \setminus Q_{00}$ and that $y$ has a neighbor $y' \in Q_{10} \setminus Q_{00}$.
As $x \nad y'$ and $y \nad x'$ and $x \ad y$, it follows that $x' \nad y'$ as otherwise $G$ contains $x \ad y \ad y' \ad x' \ad x$ as an induced $C_4$.

Now consider the vertex $u_1$, which is at the same time in $Q_{10} \setminus Q_{00}$ and in $Q_{11} \setminus Q_{00}$ according to Claim \ref{cla_uvertices}.
Hence, according to the choice of $x$ and $y$, we have $x \nad u_1$ and $y \nad u_1$.
Moreover, as $u_1$ and $x'$ are both in $Q_{11} \setminus Q_{00}$ and because $u_1$ and $y'$ are both in $Q_{10} \setminus Q_{00}$, we get $x' \ad u_1$ and $y' \ad u_1$, which implies that $G$ has $x \ad y \ad y' \ad u_1 \ad x' \ad x$ as an induced $C_5$.
This means, our assumption was wrong and we let $w_{00}$ be a vertex in $\{x,y\}$ that has no neighbors in $Q_{10} \setminus Q_{00}$ and in $Q_{11} \setminus Q_{00}$.

First of all, we have already seen that $w_{00}$ is not adjacent to $u_1$.
Therefore, $w_{00}$ is in $Q_{00} \setminus P_{00}$, as every vertex in $P_{00}$ is adjacent to $u_1$ by $P_{00} \cup \{u_1\} \subseteq T_0$.
Moreover, this means that $w_{00} \not= w_{10}$ and $w_{00} \not= w_{11}$ as both, $w_{10}$ and $w_{11}$, are adjacent to $u_1$, which follows from $\{w_{10},u_1\} \subseteq Q_{10}$ and $\{w_{11},u_1\} \subseteq Q_{11}$.
As $w_{10} \in Q_{10} \setminus Q_{00}$ and $w_{11} \in Q_{11} \setminus Q_{00}$, it follows also that $w_{00}$ is not adjacent to $w_{10}$ and $w_{11}$.

From Claim \ref{cla_pqnode:intersection:hor:dia} we know that $Q_{00} \cap Q_{10}$ and $Q_{00} \cap Q_{11}$ are subsets of $T$.
Because $P_{10} = Q_{10} \cap T_0$ and $P_{11} = Q_{11} \cap T_1$, this means also that $Q_{00} \cap P_{10} = Q_{00} \cap Q_{10} \cap T_0 \subseteq T$ and $Q_{00} \cap P_{11} = Q_{00} \cap Q_{11} \cap T_1 \subseteq T$.
Hence, from $P'_{10} = P_{10} \setminus T$ and $P'_{11} = P_{11} \setminus T$ it follows already that $w_{00}$ is not adjacent to vertices in $P'_{10}$ or in $P'_{11}$.
It remains to show that $w_{00} \not= w_{01}$, that $w_{00} \nad w_{01}$ and that $w_{00}$ is not adjacent to any vertex in $P'_{01}$.

If $W = N(w_{00}) \cap P'_{01}$ is an empty set, then $w_{00}$ is not adjacent to vertices in $P'_{01}$.
Otherwise, if $W$ is not empty, assume that there are vertices $x \in W$ and $y \in P'_{00}$ such that $x \nad y$.
Recall that $w_{00}$ is adjacent to all vertices in $P'_{00}$ including $y$ and not adjacent to $u_1$.
Unlike $w_{00}$, the vertices $x$ and $y$ are adjacent to $u_1$, because $\{u_1,y\} \subseteq T_0$ and $\{u_1,x\} \subseteq T_1$.
This means that $G$ has $w_{00} \ad x \ad u_1 \ad y \ad w_{00}$ as an induced $C_4$, a contradiction.

Consequently, $W \cup P'_{00}$ is a clique in $G$ and there is a maximal clique of $G$ represented by a sink $T'$ of $\cA(G)$ such that $(W \cup P'_{00} \cup \{u_1\}) \subseteq T'$.
Because $T = T_0 \cap T_1$, it follows from Lemma \ref{lma_node_to_sink_intersection} that there are distinct sinks $T'_0$, reachable from $T_0$, and $T'_1$, reachable from $T_1$, such that $T = T'_0 \cap T'_1$.
Clearly, $T'_0 \not= T'_1$ and because $u_1 \in T'_0$ and $u_1 \in T'_1$, we have $T' \not= T'_0$ and $T' \not= T'_1$.
We let $X = T' \cap T'_0$ and $Y = T' \cap T'_1$ and obtain a $3$-cycle with starters $X,Y,T$ and terminals $T',T'_0,T'_1$.
By construction, there are vertices $x \in P'_{00}$ and $y \in P'_{01}$ that are also contained in $T'$.
Because $P'_{00} \subseteq T'_0$ and $P'_{01} \subseteq T'_1$, it follows that $x \in T'_0$ and $y \in T'_1$ and this in turn means that $x \in X$ and $y \in Y$.
Consequently, $X \not\subseteq T'_1$, as otherwise $x \in T'_0 \cap T'_1 = T$, which is a contradiction to the construction $P'_{00} = P_{00} \setminus T$.
Analogously, we have $Y \not\subseteq T'_0$.
Finally, $T \subseteq T'$ implies that all vertices in $T$, including $u_1$, are adjacent to $w_{00}$, which is impossible.
This means that the $3$-cycle is bad and, hence, $W$ has to be empty and $w_{00}$ is not adjacent to any vertex in $P'_{00}$. 

As $w_{01}$ is adjacent to all vertices in $P'_{01}$, it follows that $w_{00} \not= w_{01}$.
Assume that $w_{00}$ and $w_{01}$ are adjacent and select any vertices $x \in P'_{00}$ and $y \in P'_{01}$.
If $x \ad y$, then we obtain $w_{00} \ad x \ad y \ad w_{01} \ad w_{00}$ as an induced $C_4$ in $G$, and otherwise, we get $w_{00} \ad x \ad u_1 \ad y \ad w_{01} \ad w_{00}$ as an induced $C_5$ in $G$.
Hence, $w_{00} \ad w_{01}$ cannot be true.
\end{proof}

\subsubsection*{The Proof of Claim \ref{cla_tsinks}}

\begin{proof}
Let $R_{00}$ be a sink of $\cA(G)$ that represents one of the maximal cliques of $G$ with $C_{00} \subseteq R_{00}$.
% If $r_{00}$ belongs to the short cycle's path from $s_{i'}$ to $t_{j'}$ for some $i',j' \in \{0,1\}$, then $r_{00}$ could only coincide with $t_{j'}$ because as a sink it has outdegree zero.
% But then it follows that $p_{0(1-j')} \subset t_{j'}$ and as $p_{0(1-j')} \subset t_{1-j'}$ anyway, this implies that $p_{0(1-j')} \subseteq t$, which is a contradiction to the construction of $p_{0(1-j')}$.
% Hence, $r_{00}$ is not situated on the short cycle.
Consider the node $T'_1$ that results from the intersection $R_{00} \cap T_1$, which exists by Lemma \ref{obs_node_intersection}, as both, $R_{00}$ and $T_1$, contain $P_{01} \cup P_{11}$.
If $T'_1 \subset T_1$, we get a new $2$-cycle with starters $S_0,S_1$ and terminals $T_0,T'_1$.
Assume that there is a node $X$ with $S_0 \cup S_1 \subseteq X \subseteq T_0 \cap T'_1$ on one of $B_{00}$, $B_{10}$, $B'_{01} = S_0 \dad \dots \dad P_{01} \dad \dots \dad T'_1$, and $B'_{11} = S_1 \dad \dots \dad P_{11} \dad \dots \dad T_1'$.
If $X$ is neither on the path $P_{01} \dad \ldots \dad T'_1$ nor on the path $P_{11} \dad \ldots \dad T'_1$, then $X$ is located on one of the paths $B_{00}$, $B_{01}$, $B_{10}$, and $B_{11}$, a contradiction to the choice of these paths.
If $P_{01} \dad \ldots \dad X \dad \ldots \dad T'_1$, then it follows that $P_{01} \subseteq X \subseteq T_0 \cap T'_1 \subset T$, which is impossible due to the construction of $P_{01}$.
The same holds if $X$ is on the path $P_{11} \dad \ldots \dad T'_1$.
Hence, the new $2$-cycle is bad.
This is a contradiction to the choice of the primal cycle, because, by $T'_1 \subset T_1$, we have $|T_0|+|T_1'| < |T_0|+|T_1|$.
Consequently, $T'_1$ equals $T_1$, and thus, $T_1 \subset R_{00}$, which means that $T_1$ is not a sink in $\cA(G)$.
\end{proof}

\subsubsection*{The Proof of Claim \ref{cla_twvertices}}

\begin{proof}
Lemma \ref{lma_node_to_sink_intersection} implies the existence of two distinct sinks $T'_0$, reachable from $T_0$, and $T'_1$, reachable from $T_1$, such that $T'_0 \cap T'_1 = t$.
Because $C_{00}$ is a clique, Claim \ref{cla_tsinks} implies that $T_1$ is not a sink, hence, $T'_1 \not= T_1$.
From Lemma \ref{lma_vertex_node_none_adjacency} it follows that $T'_1 \setminus T'_0$ contains at least one vertex $w_1$ that is not adjacent to any vertex in $T'_0 \setminus T'_1$.

As $C_{00}$ is a clique, all vertices in $P_{01}$ and in $P_{11}$ are adjacent to all vertices in $P_{00}$.
Consequently, $w_1$ is not in $P_{01} \cup P_{11}$.
If $w_0$ exists, then it can neither be the same vertex as $w_1$ nor be adjacent to $w_1$, because $w_0 \in T'_0 \setminus T'_1$.
Clearly, by Claim \ref{cla_pwvertices}, $w_1$ is not one of the vertices $w_{00},w_{10}$, because, unlike $w_1$, they are adjacent to vertices in $T_0 \setminus T$.
Moreover, Claim \ref{cla_pwvertices} implies that $w_1$ is not $w_{01}$, because, unlike $w_{01}$, the vertex $w_1$ is adjacent to all vertices in $P'_{11}$.
Similarly $w_1$ is not $w_{11}$.

It remains to show that $w_1$ is not adjacent to $w_{00}$ and $w_{10}$ and adjacent to at most one vertex $w_{10}$ or $w_{11}$.
If $w_1 \ad w_{00}$, then we can select any vertex $x \in P'_{01}$ and get $w_1 \ad u_1 \ad x \ad w_{00} \ad w_1$ as an induced $C_4$ in $G$. 
Analogously, if $w_1 \ad w_{10}$, then we select $x \in P'_{10}$ to find $w_1 \ad u_0 \ad x \ad w_{10} \ad w_1$ as induced $C_4$ in $G$.
Finally, if $w_1$ is adjacent to $w_{01}$ and $w_{11}$, then we select $x \in P'_{00}$ and get an induced $3$-sun in $G$ with central clique $u_0,u_1,w_1$ and independent set $x,w_{11},w_{01}$.
\end{proof}

\end{document}